\documentclass[submission,copyright,creativecommons, colorlinks]{eptcs}
\providecommand{\event}{GCM 2025} 

\usepackage{iftex}

\ifpdf
  \usepackage{underscore}         
  \usepackage[T1]{fontenc}        
\else
  \usepackage{breakurl}           
\fi

\usepackage[utf8]{inputenc}
\usepackage[inline]{enumitem}
	\setlist{itemsep=0pt,topsep=.5ex}
\usepackage{tikz}                     
  \usetikzlibrary{automata,matrix,calc,chains,backgrounds,positioning
                 ,fit,decorations.pathmorphing,shapes
                 ,arrows,arrows.meta}
\usepackage{amsthm} 
\usepackage{amsmath}
\usepackage{amssymb}
\usepackage{stmaryrd}
  \DeclareFontFamily{U}{mathx}{\hyphenchar\font45}
  \DeclareFontShape{U}{mathx}{m}{n}{<-> mathx10}{}
  \DeclareSymbolFont{mathx}{U}{mathx}{m}{n}
  \DeclareMathAccent{\widebar}{0}{mathx}{"73}

\newcommand{\tableref}[1]{\hyperref[#1]{Table~\ref*{#1}}}
\newcommand{\figref}[1]{\hyperref[#1]{Figure~\ref*{#1}}}
\newcommand{\figandfigref}[2]%
	{\hyperref[#1]{Figures~\ref*{#1}} \hyperref[#1]{and~\ref*{#2}}}
\newcommand{\assref}[2]%
	{\hyperref[#1]{Assumption~\ref*{#1}}.\hyperref[#1]{\ref*{#2}}}
\newcommand{\assrefs}[3]%
	{\hyperref[#1]{Assumptions~\ref*{#1}}.\hyperref[#1]{\ref*{#2}}--\hyperref[#1]{\ref*{#1}}.\hyperref[#1]{\ref*{#3}}}
\newcommand{\tabref}[1]{\hyperref[#1]{Table~\ref*{#1}}}
\newcommand{\sectref}[1]{\hyperref[#1]{Section~\ref*{#1}}}
\newcommand{\sectsref}[1]{\hyperref[#1]{Sections~\ref*{#1}}}
\newcommand{\stepref}[1]{\hyperref[#1]{Step~\ref*{#1}}}
\newcommand{\defref}[1]{\hyperref[#1]{Definition~\ref*{#1}}}
\newcommand{\thmref}[1]{\hyperref[#1]{Theorem~\ref*{#1}}}
\newcommand{\lemmaref}[1]{\hyperref[#1]{Lemma~\ref*{#1}}}
\newcommand{\factref}[1]{\hyperref[#1]{Fact~\ref*{#1}}}
\newcommand{\factsubref}[2]%
	{\hyperref[#1]{Fact~\ref*{#1}}.\hyperref[#1]{\ref*{#2}}}
\newcommand{\exref}[1]{\hyperref[#1]{Example~\ref*{#1}}}
\newcommand{\exsref}[1]{\hyperref[#1]{Examples~\ref*{#1}}}
\newcommand{\alineref}[1]{\hyperref[#1]{line~\ref*{#1}}}
\newcommand{\alinesref}[1]{\hyperref[#1]{lines~\ref*{#1}}}
\newcommand{\remref}[1]{\hyperref[#1]{Remark~\ref*{#1}}}
\newcommand{\condref}[1]{\hyperref[#1]{Condition~\ref*{#1}}}
\newcommand{\Eqref}[1]{\hyperref[#1]{Equation~\ref*{#1}}}
\newcommand{\constrref}[1]{\hyperref[#1]{Construction~\ref*{#1}}}
\newcommand{\obsref}[1]{\hyperref[#1]{Observation~\ref*{#1}}}
\newcommand{\corref}[1]{\hyperref[#1]{Corollary~\ref*{#1}}}

\makeatletter
\@ifundefined{footref}{\newcommand{\footref}[1]{\hyperref[#1]{${}^{\ref*{#1}}$}}}{}
\makeatother

\newtheoremstyle{plainbreak}%
  {}{}%
  {\itshape}{}%
  {\bfseries}{}
  {\newline}{}
\newtheoremstyle{definitionbreak}%
  {}{}%
  {}{}%
  {\bfseries}{}
  {\newline}{}
\newtheoremstyle{remarkbreak}%
  {}{}%
  {}{}%
  {\itshape}{}
  {\newline}{}
\newtheoremstyle{plain}%
  {5pt minus 2pt plus 2pt}{5pt minus 2pt plus 2pt}%
  {}{}%
  {\bfseries}{}
  {.5em}{}
\newtheoremstyle{definition}%
  {5pt minus 2pt plus 2pt}{5pt minus 2pt plus 2pt}%
  {}{}%
  {\bfseries}{}
  {.5em}{}
\newtheoremstyle{remark}%
  {5pt minus 2pt plus 2pt}{5pt minus 2pt plus 2pt}%
  {}{}%
  {\itshape}{}
  {.5em}{}
\theoremstyle{plain}
  \newtheorem{defn}{Definition}[section]

  \newtheorem{example}{Example}[section] 

  \newtheorem{theorem}[defn]{Theorem}
  \newtheorem{lemma}[defn]{Lemma}

  \newtheorem{fact}[defn]{Fact}

\newcommand{\opname}[1]{\operatorname{\text{$#1$}}}

\def\phi{\varphi}
\def\rho{\varrho}
\def\theta{\vartheta}

\newcommand{\card}[1]{| #1 |}              
\def\emptyset{\varnothing}
\newcommand{\Nat}{\mathbb{N}}
\newcommand{\tup}[1]{\langle #1 \rangle}   
\newcommand{\true}{\mathit{true}}
\newcommand{\false}{\mathit{false}}

\newcommand{\rank}{\mathit{rank}}   
\newcommand{\Voc}{\Sigma}       

\newcommand{\emptyseq}{\varepsilon}        

\newcommand{\att}{\mathit{att}}   
\newcommand{\ed}[1]{\bar{#1}}    

\newcommand{\gcat}{\mathop{\odot}}         
\newcommand{\off}{\ominus}           
\newcommand{\cg}{\mathit{CG}}           
\newcommand{\front}{\mathit{front}}               
\newcommand{\Idg}[1]{\mathit{Id}_{#1}}   
\newcommand{\lab}{\mathit{lab}}               
\newcommand{\nd}[1]{\dot{#1}}   
\newcommand{\rear}{\mathit{rear}}               

\newcommand{\src}{\mathit{src}}   
\newcommand{\tgt}{\mathit{tgt}}   

\def\Exp{\mathbb{E}}
\def\ex{\mathit{ex}}

\def\fo{\mathit{fo}}

\def\satisfies{\vDash}                   

\newcommand{\aut}[1]{\mathfrak{#1}}
\newcommand{\trans}{\vdash}
\newcommand{\ass}{\textit{Ass}}
\newcommand{\evo}{\textit{Evo}}

\def\pg(#1){\tup{#1}} 

\newcommand{\EE}{\mathbb{E}}
\def\GG{\mathbb{G}}
\newcommand{\M}{\mathcal{M}}
\newcommand{\LL}{\mathcal{L}}

\newcommand{\invis}{\mbox{\textvisiblespace}}

\newcommand{\Ident}{X}
\newcommand{\id}{x}

\newcommand{\sign}{\mathit{sign}}
\newcommand{\RecFormula}{\mathcal{F}}
\DeclareMathOperator{\evalstep}{\triangleright}
\newcommand{\fcg}{\mathit{FCG}}

\newcommand{\spos}{\mathrm{pos}}
\newcommand{\sneg}{\mathrm{neg}}
\newcommand{\egraph}{E}
\newcommand{\Quant}{\mathsf{Q}}

\newcommand\raiseb[1]{\!\!\!\raisebox{0.6ex}{#1}\!\!}

\def\CFG#1#2#3{#1,#2,#3}
\def\CFGr#1#2#3{#1,\raiseb{#2},#3}
\pgfdeclarelayer{background}              
\pgfdeclarelayer{middleground}
\pgfdeclarelayer{foreground}
\pgfsetlayers{background,middleground,main,foreground}

\newcommand\Red{red!75!black}

\newcommand\Green{green!40!black}

\newcommand\CloudColor{gray!20!white}
\tikzset{
  x=8mm,y=8mm,>=latex,            
  background rectangle/.style
  ={fill=\BackgroundColor,rounded corners=4pt
      },
  glass/.style ={opacity=0,text opacity=0},     
  satin/.style ={opacity=0.3,text opacity=0.3},
  link/.style={-,decorate,
         decoration={snake,amplitude=0.5pt,segment length=2pt}},
  curly/.style={->,decorate,
         decoration={snake,amplitude=0.5pt,segment length=2pt,post length=2pt}},
  e/.style={inner sep=0.0pt,outer sep=0.0pt,minimum size=0pt},
  o/.style={circle,draw,solid,fill=white,font=\tiny,inner sep=1pt,minimum size=2.5mm},
   hyper/.style={rectangle,fill=white,draw,inner sep=0pt,minimum
     size=3.5mm,font=\scriptsize},
  aut/.style={rectangle,inner sep=2pt,minimum height=4mm
             ,draw,fill=\BackgroundColor},
  universal/.style={fill=black,text=white},
  config/.style={rectangle,rounded corners=4pt,draw,very thin,inner sep=2pt},
  acc/.style={config,thick,draw=\Green},
  rej/.style={config,thick,densely dashed,draw=\Red},
  Cloud/.style={cloud,cloud puffs=17,minimum width=15mm,minimum height=5mm,fill=\CloudColor,draw},
 node id/.style={label distance=2pt,inner sep=0pt,outer sep=0pt,
                  anchor=mid,text=black,font=\tiny\tt},
  blob/.style={rectangle,rounded corners=6mm,fill opacity=0.4,draw,very thin}, 
  blub/.style={rectangle,rounded corners=6mm,inner sep=2pt}, 
  meta/.style={double equal sign distance,-Implies},
  }
  {\begin{tikzpicture}%
    [x=6mm,y=-4mm, 
    #1]}%
  {\end{tikzpicture}}

\newenvironment{automaton}[1][x=10mm,y=10mm]%
{\begin{tikzpicture}[background
    rectangle/.style={fill=white,rounded
      corners=4pt,draw,thin},
   inner frame sep=4pt,show background rectangle,
    baseline=(current bounding box.center),
    node distance=8mm,
    initial text=,
    every state/.style={fill=white,minimum size=1mm,inner sep=1pt},
    every label/.style={node id},
    #1]}%
  {\end{tikzpicture}}
\newcommand\BackgroundColor{gray!20!white}
\newenvironment{graph}[1][x=10mm,y=10mm]%
{\begin{tikzpicture}[background
    rectangle/.style={fill=\BackgroundColor,rounded
      corners=4pt,draw,solid,ultra thin},
    inner frame sep=0pt,show background rectangle,
    baseline=(current bounding box.center),
    every label/.style={node id},
    #1]}%
  {\end{tikzpicture}}
\def\frontptr(#1){\path (f-#1) edge[-,double distance=1pt,double,solid] (#1);}
\def\rearptr(#1){\path (#1) edge[-,out=0,in=180,double distance=1pt,double,solid] (r-#1.east);}
\def\invnode(#1)(#2){
  \node[minimum size=8pt,glass] at (#2) {};
}
\def\inode(#1)(#2){
  \node[o,minimum size=1.5mm] (#1) at (#2) {};
  \node[minimum size=8pt] at (#2) {};
  \node[e]        at ($(#2)-(0.3,0)$) {};
  \node[e]        at ($(#2)+(0.3,0)$) {};
}
\def\Inode(#1)(#2){
  \node[o] (#1) at (#2) {$#1$};
  \node[e]        at ($(#2)-(0.3,0)$) {};
  \node[e]        at ($(#2)+(0.3,0)$) {};
}
\def\fnode(#1)(#2){
  \inode (#1)(#2);
  \node[e] (f-#1) at ($(#2)-(0.4,0)$) {};
  \frontptr(#1);
}
\def\Fnode(#1)(#2){
  \Inode (#1)(#2);
  \node[e] (f-#1) at ($(#2)-(0.4,0)$) {};
  \frontptr(#1);
}
\def\rnode(#1)(#2){
  \inode (#1)(#2);
  \node[e] (r-#1) at ($(#2)+(0.4,0)$) {};
  \rearptr(#1);
}
\def\Rnode(#1)(#2){
  \Inode (#1)(#2);
  \node[e] (r-#1) at ($(#2)+(0.4,0)$) {};
  \rearptr(#1);
}
\def\frnode(#1)(#2){
  \fnode(#1)(#2);
  \node[e] (r-#1) at ($(#2)+(0.4,0)$){};
  \rearptr(#1);
}
\def\FRnode(#1)(#2){
  \Fnode(#1)(#2);
  \node[e] (r-#1) at ($(#2)+(0.4,0)$){};
  \rearptr(#1);
}
  {\begin{tikzpicture}%
     [>=Computer Modern Rightarrow                 
     ,every node/.style={rectangle,inner sep=3pt}
     ,incl/.tip={Hooks[round,right]}               
     ,incL/.tip={Hooks[round,left]}
     ,part/.tip={Circle[open]} 
     ,weak/.tip={Latex[open]} 
     ,harp/.tip={Computer Modern Rightarrow[left]} 
     ,Harp/.tip={Computer Modern Rightarrow[right]}
     ,to/.style={->}
     ,ito/.style={incl->}
     ,ifrom/.style={incL->}
     ,eto/.style={->>}
     ,mto/.style={>->}
     ,pto/.style={part->}
     ,peto/.style={part->>}
     ,pmto/.style={part>->}
     ,tow/.style={-weak}
     ,etow/.style={->weak}
     ,mtow/.style={>-weak}
     ,ptow/.style={part->weak}
     ,petow/.style={part->weak}
     ,pmtow/.style={part>-weak}
     ,To/.style={double equal sign distance,-Implies}
     ,x=12mm,y=12mm
     ,#1]}%
  {\end{tikzpicture}}
  {\begin{tikzpicture}[#1]}%
  {\end{tikzpicture}}
\makeatletter

\def\DownEdge[#1]{\tikz \draw (0pt,8pt) edge[#1] (0pt,0pt);}

\def\blob(#1) at (#2)#3{%
  \node[blub](#1)at(#2) {\begin{tabular}{@{}c@{}}  #3 \end{tabular}
};
  }
\def\enlargebb{\node [glass,fit= (current bounding box),inner sep=2pt] {};}

\def\Highlight#1#2#3(#4){
  \begin{pgfonlayer}{background} 
      \node (#4) [subgraph,fit= #1,inner sep=#3,fill=#2] {}; 
   \end{pgfonlayer}
}

\def\Just#1%
  {\tikz \draw (0pt,0pt) edge[#1] (12pt,0pt); 
   \end{graph}}
\def\Graph{\@ifnextchar[{\@Graph}{\@GrapH}}
\def\@Graph[#1]#2{%
  \BOX{%
    \begin{tikzpicture}[x=8mm,y=8mm,label distance=-2pt,>=latex,#1]
     #2%
    \end{tikzpicture}%
  } 
}
\def\@GrapH#1{%
  \BOX{%
    \begin{tikzpicture}[x=8mm,y=8mm,label distance=-2pt,>=latex]
     #1%
    \end{tikzpicture}%
  } 
}

\def\proGraph{\@ifnextchar[{\@proGraph}{\@pro@Graph}}
\def\@proGraph[#1]#2{\BOX{%
  \begin{tikzpicture}[x=8mm,y=-7mm,>=latex,every label/.style={elab},#1]
     #2
    \enlargebb
   \end{tikzpicture}}}
\def\@pro@Graph#1{\BOX{%
  \begin{tikzpicture}[x=8mm,y=-7mm,>=latex,every label/.style={elab}]
    #1
    \enlargebb
  \end{tikzpicture}}}

\def\uvar(#1)#2{\@ifnextchar[{\@uvar(#1)#2}{%
  \node (#1-node) at ($(#1)+(0,1)$) {#2};
  \path[arm] (#1) edge (#1-node);
}}
\def\@uvar(#1)#2[#3]{%
  \node (#1-node) at ($(#1)+(#3,1)$) {#2};
  \path[arm] (#1) edge (#1-node);
}
\def\ustar#1#2{\@ifnextchar[{\@ustar{#1}{#2}}{\Graph{%
      \node (r) [term] at (1,1) {#1};
      \node (h) [nont] at (1,2) {#2};
      \path (r) edge[arm] (h);
}}}
\def\@ustar#1#2[#3]{\Graph[#3]{%
      \node (r) [term] at (1,1) {#1};
      \node (h) [nont] at (1,2) {#2};
      \path (r) edge[arm] (h);
}}

\def\custar#1#2#3{\@ifnextchar[{\@custar(#1)#2}{\proGraph{%
      \node (r) [term] at (1,1) {#1};
      \node (h) [nont] at (1,2) {#2};
      \path (r) edge[arm] (h);
      \node (c) [term] at (1,3) {#3};
}}}
\def\@custar#1#2#3[#4]{\proGraph[#4]{%
      \node (r) [term] at (1,1) {#1};
      \node (h) [nont] at (1,2) {#2};
      \path (r) edge[arm] (h);
      \node (c) [term] at (1,3) {#3};
}}
\makeatother

\newcommand{\BOX}[1]{\begin{array}{@{}c@{}}#1\end{array}}

\title{Systems of Graph Formulas and their\\Equivalence to Alternating
  Graph Automata}
\author{Frank Drewes\thanks{Partially supported by the Swedish Research Council under grant no. 2024-05318, and by the Wallenberg AI, Autonomous Systems and Software Program through the NEST project STING — Synthesis and Analysis with Transducers and Invertible Neural Generators.}
  \institute{
    Ume\aa\ universitet \\
    Ume\aa, Sweden}
  \email{drewes@cs.umu.se}
  \and Berthold Hoffmann
  \institute{
    Universität Bremen\\
    Bremen, Germany}
  \email{hof@uni-bremen.de}
  \and Mark Minas
  \institute{
    Universität der Bundeswehr München \\
    Neubiberg, Germany}
  \email{mark.minas@unibw.de}
}
\newcommand{\authorrunning}{F. Drewes, B. Hoffmann, M. Minas}
\newcommand{\titlerunning}{Systems of Graph Formulas and their Equivalence to Alternating
  Graph Automata}

\begin{document}
\maketitle
\begin{abstract}
   Graph-based modeling plays a fundamental role in many areas of computer
   science. In this paper, we introduce systems of graph formulas with
   variables for specifying graph properties; this notion generalizes
   the graph formulas introduced in earlier work by incorporating recursion. We
   show that these formula systems have the same expressive power as alternating
   graph automata, a computational model that extends traditional finite-state
   automata to graphs, and allows both existential and universal states. In
   particular, we provide a bidirectional translation between formula systems
   and alternating graph automata, proving their equivalence in specifying graph
   languages. This result implies that alternating graph automata can be
   naturally represented using logic-based formulations, thus bridging the gap
   between automata-theoretic and logic-based approaches to graph language
   specification. 
\end{abstract}

\section{Introduction\label{s:intro}}
Graph-like structures are ubiquitous in computer science and beyond.
In many cases, they need to fulfill application-dependent structural
properties that have to be specified and checked.  In the literature,
such properties have been called
global~\cite{Gaifman:82} if they concern graph regions of unbounded
size, like connectedness or the existence of cycles.
In this paper, we relate two mechanisms for specifying global graph properties:
systems of graph formulas and  alternating graph automata.

  Graph expressions based on the notion of graph concatenation by Engelfriet
  and Vereijken~\cite{Engelfriet-Vereijken:97} have been introduced
  in~\cite{Drewes-Hoffmann-Minas:24,derosa-minas:24}. They resemble ordinary regular expressions
  and define a class of graph languages which can be shown to be equal to the
  class of graph languages accepted by finite graph automata. In~\cite{dhm-icgt25b}, we
  have used graph expressions to introduce a more powerful notion of graph formulas.
  Rather than being a special case of formulas in predicate logic such as monadic second-order
  logic~\cite{Courcelle-Engelfriet:12}, the
  structure of our formulas corresponds to that of the nested graph conditions
  by Habel and Pennemann~\cite{Habel-Pennemann:05,Pennemann:09}.
  However, the basic
  building blocks are graph expressions in the former case rather than graph morphisms as in the latter case. Graph
  properties specified by graph formulas are in the polynomial hierarchy PH.
  
  On the one hand, graph formulas can be translated to the alternating graph automata
  introduced in~\cite{dhm-icgt25a}. This was shown in~\cite{dhm-icgt25b}, thus proving
  that the latter are at least as powerful as the former. On the other hand,
  from~\cite{dhm-icgt25a} alternating graph automata are known to be capable of
  describing PSPACE-complete graph properties. Together with the aforementioned fact
  that graph formulas can only capture graph properties in the polynomial hierarchy,
  this shows that alternating graph automata are strictly more powerful than graph
  formulas, provided that $\mathrm{PH} \neq \mathrm{PSPACE}$.
  
  Intuitively, the reason for this difference in power is that each graph formula
  has a fixed number of quantifier alternations as each formula syntactically
  resembles a tree. In contrast, alternating graph automata can be cyclic, thus
  implementing unbounded quantifier depth.
  
  In this paper we extend graph formulas to systems of graph formulas with
  variables, similar to Flick who extended nested graph conditions to
  recursively nested graph conditions~\cite{Flick:16}. Also similar to Mezei and
  Wright's notion of systems of language equations~\cite{Mezei-Wright:67},
  such a system consists of a finite set of variables, each of which is
  mapped to a graph formula. Intuitively, each variable represents the set of
  graphs that satisfy the associated formula. Since every formula may itself
  contain variables, the system may be cyclic. This leads us to define its
  semantics as a least fixed point.

  After preliminary definitions in \sectref{s:prel} and a recap of alternating
  graph automata in \sectref{s:afga}, we define systems of graph formulas and
  their semantics in \sectref{s:graph formulas}, where we also show that this
  semantics coincides with that of graph formulas in the acyclic case
  (\lemmaref{l:inductive-semantics}). Furthermore, we establish a useful normal
  form of systems of graph formulas. Using this normal form, we finally prove
  the main result of this paper in \sectref{s:power}: systems of graph formulas
  are precisely as powerful as alternating graph automata (\thmref{thm:main}).
  As a running example, we specify the language of all graphs that have a
  Hamiltonian cycle, which cannot be specified by recursively nested
  conditions~\cite[p.~143]{Flick:16}.

\section{Preliminaries}\label{s:prel}

We let $\Nat$ denote the set of non-negative integers and $[n]$ the
set $\{1,\dots,n\}$ for all $n\in\Nat$. $A^*$~denotes the set of all
finite sequences over a set~$A$; the empty sequence is denoted by
$\emptyseq$. 
For a sequence $s\in A^*$, $[s]$ denotes the set
of all members of $A$ occurring in~$s$. 
%
For a binary relation
${\leadsto}\subseteq A\times B$, we write $\not\leadsto$ for its complement, 
i.e., for $a\in A$ and $b\in B$, $a\not\leadsto b$ holds if and only 
if $a\leadsto b$ does not.

We consider edge-labeled hypergraphs (which we simply call graphs),
i.e., edges are attached to
sequences of nodes.  To be able to concatenate graphs in the way originally
proposed by Engelfriet and Vereijken~\cite{Engelfriet-Vereijken:97}, we supply
each graph with two
sequences of distinguished nodes, its front and rear
interfaces.  As in
\cite{Drewes-Habel-Kreowski:97}, we require these sequences to be free
of repetitions.%

A ranked set $(\mathcal{S},\rank)$ consists of
a finite set $\mathcal{S}$ of elements and a function
$\rank \colon \mathcal{S} \to \Nat$, which assigns
a rank to each element $a\in\mathcal{S}$. The pair $(\mathcal{S},\rank)$
is usually identified with $\mathcal{S}$, keeping $\rank$ implicit.
For $k\in\Nat$, we let $\mathcal{S}^{(k)}=\{a\in\mathcal{S}\mid\rank(a)=k\}$.

\begin{defn}[Graph]\label{s:graph}%
  Let $\Voc$ be a ranked set of symbols. A \emph{graph} over $\Voc$ is a tuple
  $G = (\nd G, \ed G, \att_G, \allowbreak \lab_G, \allowbreak
    \front_G, \allowbreak \rear_G)$, where $\nd G$ and $\ed G$ are
  disjoint finite sets of \emph{nodes} and \emph{edges}, respectively,
  $\att_G \colon \ed G \to \nd G ^*$ attaches sequences
  of nodes to edges,
  $\lab_G \colon \ed G \to \Voc$ labels edges with symbols so
  that $\card{\att_G(e)} = \rank(\lab_G(e))$ for every edge
  $e \in \ed G$, and the repetition-free node sequences
  $\front_G, \rear_G \in \nd G^*$ are the
  \emph{front} and the \emph{rear} interface, respectively.

  The \emph{type} of a graph $G$ is $(\card{\front_G},\card{\rear_G})$.
  The set of all graphs of type $(m,n)$ is denoted by $\GG_\Voc^{(m,n)}$.
  Furthermore, $\GG_\Voc^m = \bigcup_{n \in\Nat} \GG_\Voc^{(m,n)}$
  and $\GG_\Voc = \bigcup_{m,n \in\Nat} \GG_\Voc^{(m,n)}$. A subset $\LL
  \subseteq \GG_\Voc$ is called a \emph{graph language}, and if
  $\LL\subseteq\GG_\Voc^{(m,n)}$, then it is a graph language of type
  $(m,n)$.
\end{defn}

A \emph{permutation graph} is a graph $G$ with $\ed G=\emptyset$ and
$\nd G=[\front_G]=[\rear_G]$. If, furthermore, $\front_G=\rear_G$ and
$|\nd G|=n$, then $G$ is an \emph{identity graph} (on $n$ nodes)
and is denoted by $\Idg n$.

\def\TreeWithout(#1)(#2)#3#4{
  \begin{graph}[x=1.5mm,y=-4mm]
    \node[minimum size=8pt] (g-1) at (1,0) {}; 
    \node[minimum size=8pt] (g-2) at (0,1) {};
    \node[minimum size=8pt] (g-3) at (2,1) {};
    \node[minimum size=8pt] (g-4) at (0,2) {};
    \node[e] (f-1) at (-1.5,0.5) {};  
    \node[e] (f-2) at (-1.5,1.5) {};
    \path (f-1) edge[overlay,-,double distance=1pt,double,out=0,in=180] (#1);
    \path (f-2) edge[overlay,-,double distance=1pt,double,out=0,in=180] (#2);
    \foreach \nd in {#3} {\node[o,minimum size=1.5mm] (\nd) at ($(g-\nd)$) {};}
    \foreach \src / \tgt in {#4} {\path (\src) edge[->] (\tgt);}
  \end{graph}
}
\def\TreeAlone#1#2{
  \begin{graph}[x=1.5mm,y=-4mm]
    \node[minimum size=8pt] (g-1) at (1,0) {}; 
    \node[minimum size=8pt] (g-2) at (0,1) {};
    \node[minimum size=8pt] (g-3) at (2,1) {};
    \node[minimum size=8pt] (g-4) at (0,2) {};
    \foreach \nd in {#1} {\node[o,minimum size=1.5mm] (\nd) at ($(g-\nd)$) {};}
    \foreach \src / \tgt in {#2} {\path (\src) edge[->] (\tgt);}
  \end{graph}
}
\def\GoutEN{
  \begin{graph}[x=5mm,y=-5mm]
    \fnode(Ef-x)(1,1)
    \rnode(Ef-z)(2,1)
    \fnode(Ef-y)(1,2)
    \node[e] (r-Ef-y) at (2.4,2) {};
    \rearptr(Ef-y)
    \path
    (Ef-x) edge[->] (Ef-z)
    ;
  \end{graph}
}
\newsavebox\BoutEN \sbox\BoutEN{\GoutEN}

\begin{defn}[Graph Concatenation \cite{Engelfriet-Vereijken:97}]\label{d:graphops}
  Let $G \in \GG_\Voc^{(i,k)}$ and $H \in \GG_\Voc^{(k,j)}$.  We
  assume for simplicity that $\rear_G=\front_H$,
  $\nd G\cap\nd H=[\rear_G]=[\front_H]$, and
  $\ed G\cap\ed H=\emptyset$. (Otherwise, appropriate isomorphic copies
  of $G$ and $H$ are used.)  The (\emph{typed}) \emph{concatenation}
  $G \cdot H$ of $G$ and $H$ is the graph $C$ such that
  $\nd C=\nd G\cup\nd H$, $\ed C=\ed G\cup\ed H$,
  $\att_C=\att_G\cup\att_H$, $\lab_C=\lab_G\cup\lab_H$,
  $\front_C=\front_G$, and $\rear_C=\rear_H$. Thus,
  $C \in \GG_\Voc^{(i,j)}$.
\end{defn}

Note that $G\cdot H$ is defined on concrete graphs if
the assumptions in \defref{d:graphops} are satisfied, but is otherwise only
defined up to isomorphism. To avoid
unnecessary technicalities, we usually assume without mentioning that the former
is indeed the case.

We will need to find frontal subgraphs of a graph and to cut them off in order
to define the semantics of alternating graph automata and systems of
graph formulas in the next sections.

\begin{defn}[Frontal Subgraphs and Cutting Them Off]\label{d:frontal-subgraph}
  A graph $G\in \GG_\Voc^{(i,k)}$ is a \emph{frontal subgraph} of $H \in \GG_\Voc^{(i,j)}$
  if
  \begin{itemize}
    \item $\nd G\subseteq\nd H$,
    \item $\ed G\subseteq\ed H$ with $\lab_G(e)=\lab_H(e)$ and
          $\att_G(e)=\att_H(e)$ for all $e\in\ed G$,
    \item $\front_G=\front_H$, and
    \item for all $v\in\nd G$, $v\in[\rear_H]$ implies $v\in[\rear_G]$.
  \end{itemize}
  Then, \emph{cutting $G$ off $H$} yields the graph $H \off G$, which is obtained from $H$ by 
  removing all nodes in $\nd G\setminus[\rear_G]$, as well as 
  all edges that are either in $\ed G$ or attached to a node in $\nd G\setminus[\rear_G]$. 
  In the resulting graph, $\rear_G$ becomes the front interface and $\rear_H$ becomes the rear interface.
\end{defn}

\def\TreeGx{
	\begin{graph}[x=1.8mm,y=-5.5mm]
		\node[e] (f-1) at (-1.5,0) {};
		\node[e] (f-2) at (-1.5,2) {};
		\node[e] (r-1) at (3.5,2) {};
		\node[o,minimum size=3mm] (1) at (1,0) {a};
		\node[o,minimum size=3mm] (2) at (1,1) {b};
		\node[o,minimum size=3mm] (3) at (0,2) {c};
		\node[o,minimum size=3mm] (4) at (2,2) {d};
		\node[e,minimum size=4mm] at (1,0) {};
		\node[e,minimum size=4mm] at (2,2) {};
		\path
		(1) edge[overlay,double distance=1pt,double] (f-1)
		(3) edge[overlay,double distance=1pt,double] (f-2)
		(4) edge[overlay,double distance=1pt,double] (r-1);
		\path[->]
		(1) edge (2)
		(2) edge (3) edge (4);
	\end{graph}
}
\def\TreeGy{
	\begin{graph}[x=1.8mm,y=-5.5mm]
		\node[e] (f-1) at (-1.5,0) {};
		\node[e] (f-2) at (-1.5,2) {};
		\node[e] (r-1) at (3.5,1) {};
		\node[o,minimum size=3mm] (1) at (1,0) {a};
		\node[o,minimum size=3mm] (2) at (0,1) {b};
		\node[o,minimum size=3mm] (3) at (0,2) {c};
		\node[o,minimum size=3mm] (4) at (2,1) {d};
		\node[e,minimum size=4mm] at (1,0) {};
		\node[e,minimum size=4mm] at (2,2) {};
		\path
		(1) edge[overlay,double distance=1pt,double] (f-1)
		(3) edge[overlay,double distance=1pt,double] (f-2)
		(4) edge[overlay,double distance=1pt,double] (r-1);
		\path[->]
		(1) edge (2) edge (4)
		(2) edge (3);
	\end{graph}
}
\def\FrontalSubGraph(#1){
	\begin{graph}[x=1.8mm,y=-5.5mm]
		\node[e] (f-1) at (-1.5,0) {};
		\node[e] (f-2) at (-1.5,2) {};
		\node[e] (r-1) at (1.5,1) {};
		\node[e] (r-2) at (1.5,2) {};
		\node[o,minimum size=3mm] (1) at (0,0) {a};
		\node[o,minimum size=3mm] (2) at (0,1) {#1};
		\node[o,minimum size=3mm] (3) at (0,2) {c};
		\node[e,minimum size=4mm] at (0,0) {};
		\node[e,minimum size=4mm] at (0,2) {};
		\path
		(1) edge[overlay,double distance=1pt,double] (f-1)
		(2) edge[overlay,double distance=1pt,double] (r-1)
		(3) edge[overlay,double distance=1pt,double] (f-2)
		edge[overlay,double distance=1pt,double] (r-2);
		\path[->]
		(1) edge (2);
	\end{graph}
}
\def\CutResultX{
	\begin{graph}[x=1.8mm,y=-5.5mm]
		\node[e] (f-1) at (-1.5,1) {};
		\node[e] (f-2) at (-1.5,2) {};
		\node[e] (r-1) at (3.5,2) {};
		\node[o,minimum size=3mm] (2) at (1,1) {b};
		\node[o,minimum size=3mm] (3) at (0,2) {c};
		\node[o,minimum size=3mm] (4) at (2,2) {d};
		\node[e,minimum size=4mm] at (1,0) {};
		\node[e,minimum size=4mm] at (2,2) {};
		\path
		(2) edge[overlay,double distance=1pt,double] (f-1)
		(3) edge[overlay,double distance=1pt,double] (f-2)
		(4) edge[overlay,double distance=1pt,double] (r-1);
		\path[->]
		(2) edge (3) edge (4);
	\end{graph}
}
\def\CutResultY{
	\begin{graph}[x=1.8mm,y=-5.5mm]
		\node[e] (f-1) at (-1.5,1) {};
		\node[e] (f-2) at (-1.5,2) {};
		\node[e] (r-1) at (3.5,1) {};
		\node[o,minimum size=3mm] (2) at (0,1) {b};
		\node[o,minimum size=3mm] (3) at (0,2) {c};
		\node[o,minimum size=3mm] (4) at (2,1) {d};
		\node[e,minimum size=4mm] at (2,0) {};
		\node[e,minimum size=4mm] at (0,2) {};
		\path
		(2) edge[overlay,double distance=1pt,double] (f-1)
		(3) edge[overlay,double distance=1pt,double] (f-2)
		(4) edge[overlay,double distance=1pt,double] (r-1);
		\path[->]
		(2) edge (3);
	\end{graph}
}
\def\CutResultZ{
	\begin{graph}[x=1.8mm,y=-5.5mm]
		\node[e] (f-1) at (-2,0.4) {};
		\node[e] (f-2) at (-2,2) {};
		\node[e] (r-1) at (3.5,1) {};
		\node[o,minimum size=3mm] (2) at (0,1) {b};
		\node[o,minimum size=3mm] (3) at (0,2) {c};
		\node[o,minimum size=3mm] (4) at (2,1) {d};
		\node[e,minimum size=4mm] at (1,0) {};
		\node[e,minimum size=4mm] at (2,2) {};
		\path
		(4) edge[overlay,double distance=1pt,double,out=125,in=0] (f-1)
		(3) edge[overlay,double distance=1pt,double] (f-2)
		(4) edge[overlay,double distance=1pt,double] (r-1);
		\path[->]
		(2) edge (3);
	\end{graph}
}

\begin{example}\label{x:front-cut}
	\figref{f:front-cut} shows graphs $G$ and $G'$ of type $(2,1)$, from which
	frontal subgraphs $F$ and $F'$ of type $(2,2)$ are cut off, yielding the
	results $R$, $R'$, and $R''$. Nodes with the same name (written inside the
	circle)	correspond to each
	other in subgraph relations.
	
\begin{figure}[h]
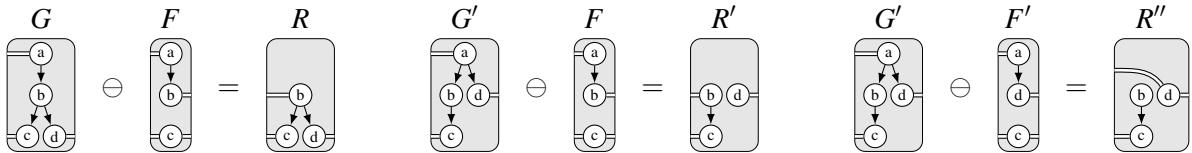

	\centering
	$
		\begin{array}{@{}*{17}{c}@{}}
			G       &             & F                   &   & R
			        &             &
			G'      &             & F                   &   & R'
			        &             &
			G'      &             & F'                  &   & R''         \\
			\TreeGx & \off        & \FrontalSubGraph(b) & = & \CutResultX
			        & \hspace*{1.5em} &
			\TreeGy & \off        & \FrontalSubGraph(b) & = & \CutResultY
			        & \hspace*{1.5em} &
			\TreeGy & \off        & \FrontalSubGraph(d) & = & \CutResultZ
		\end{array}
	$
	\caption{Cutting off frontal subgraphs.}
	\label{f:front-cut}
\end{figure}
	Our graphical conventions for drawing graphs are as follows. Binary edges
	like the ones in this example are drawn as arrows. Edges of other ranks are
	drawn as boxes connected to their attached nodes by lines.
	If there is more than one edge label, labels are
	ascribed to the arrows or written inside the boxes, or we may distinguish
	differently labeled edges by drawing them in different styles.
	Throughout the paper, we reserve the binary special edge label $\invis$,
	the	``invisible label'', to mean ``unlabeled''. Graphs over $\{\invis\}$ are
	called \emph{unlabeled graphs}.
	
	Front and rear interface nodes are connected with double lines to the left 
	and right border, respectively, of the graph. Both sequences are ordered from
	top to bottom. The graph $F$, for instance, thus has type $(2,2)$.

	Note that even though $F$ and $F'$ are isomorphic,
	they are different frontal subgraphs of $G'$, leading to different results
	$R'=G'\off F$ and $R''=G'\off F'$. Another fact worth noting is that cutting $F$
	and $F'$ off $G'$
	also removes the other edge attached to node $a$, which would otherwise dangle.
	Thus, $\off$ is not the inverse of graph concatenation.
\end{example}

A basic property of the cut operation is that it is compatible with graph concatenation:
\begin{fact}\label{f:off-gcat}
  For all graphs $G$, $\Gamma$, $\Gamma'$ of appropriate types, it holds that
  $(G \off\Gamma ) \off \Gamma' = G \off (\Gamma \cdot \Gamma')$.
\end{fact}

\newcommand{\UB}[1]{\usebox{#1}}
\def\Gloop{
  \begin{graph}[x=5mm,y=-5mm]
    \inode(x)(1,1)
    \invnode(y)(1,2)
    \path
    (x) edge[->,loop below] ()
    ;
  \end{graph}
}
\newsavebox\Bloop \sbox\Bloop{\Gloop}
\def\GedgeU{
  \begin{graph}[x=5mm,y=-5mm]
    \rnode(Up-x)(1,1)
    \rnode(Up-y)(1,2)
    \path
    (Up-y) edge[->] (Up-x)
    ;
  \end{graph}
}
\newsavebox\BedgeU \sbox\BedgeU{\GedgeU}
%
\def\GoutE{
  \begin{graph}[x=5mm,y=-5mm]
    \fnode(Ef-x)(1,1)
    \rnode(Ef-z)(2,1)
    \path
    (Ef-x) edge[->] (Ef-z)
    ;
  \end{graph}
}
\newsavebox\BoutE \sbox\BoutE{\GoutE}
%
\def\Guncurlyedge{
  \begin{graph}[x=5mm,y=-5mm]
    \fnode(Gu-x)(1,1)
    \inode(Gu-y)(2,1)
    \path
    (Gu-x) edge[curly] (Gu-y)
    ;
  \end{graph}
}
%
%
\def\GedgeD{
  \begin{graph}[x=5mm,y=-5mm]
    \fnode(Do-z)(1,1)
    \fnode(Do-y)(1,2)
    \path
    (Do-z) edge[->] (Do-y)
    ;
  \end{graph}
}
\newsavebox\BedgeD \sbox\BedgeD{\GedgeD}
%
\def\Gnode{
  \begin{graph}[x=5mm,y=-2.5mm]
    \inode(x)(1,1.5)
    \node[glass,minimum size=8pt] at (1,1) {};
    \node[glass,minimum size=8pt] at (1,2) {};
    ;
  \end{graph}
}
\newsavebox\Bnode \sbox\Bnode{\Gnode}
%

\def\Gempty{
  \begin{graph}[x=5mm,y=-5mm]
    \invnode(Idz-x)(1,1)
    \invnode(Idz-y)(1,2)
  \end{graph}
}
\newsavebox{\Bempty} \sbox\Bempty{\Gempty}
\def\Grnode{
  \begin{graph}[x=5mm,y=-5mm]
    \rnode(N-x)(1,1)
    \invnode(N-y)(1,2)
  \end{graph}
}
\newsavebox{\Brnode}   \sbox\Brnode{\Grnode}
\def\GedgeB{
  \begin{graph}[x=5mm,y=-5mm]
    \fnode(I-x)(1,1)
    \inode(I-y)(1,2)
    \path
    (I-y) edge[->] (I-x)
    ;
  \end{graph}
}
\newsavebox{\BedgeB} \sbox \BedgeB{\GedgeB}
\def\Gnodenode{
  \begin{graph}[x=5mm,y=-5mm]
    \frnode(NN-x)(1,1)
    \rnode(NN-y)(1,2)
  \end{graph}
}
\newsavebox{\Bnodenode}
\sbox\Bnodenode{\Gnodenode}
\def\GinEE{
  \begin{graph}[x=2mm,y=-5mm]
    \begin{scope}
      \inode(II-x)(1.5,1)
      \inode(II-y)(1,0)
      \inode(II-z)(2,0)
      \path[->] (II-y) edge (II-x) (II-z) edge (II-x)
      ;
    \end{scope}
  \end{graph}
}
\newsavebox{\BinEE} \sbox\BinEE{\GinEE}
\def\GoutEEE{
  \begin{graph}[x=4.5mm,y=-5mm]
    \inode(OOO-x)(1.5,1)
    \inode(OOO-y)(1,2)
    \inode(OOO-z)(1.5,2)
    \inode(OOO-w)(2,2)
    \path
    (OOO-x) edge[->] (OOO-y)  edge[->] (OOO-z) edge[->] (OOO-w)
    ;
  \end{graph}
}
\newsavebox{\BoutEEE}
\sbox\BoutEEE{\GoutEEE}
\def\Gparedge{
  \begin{graph}[x=5mm,y=-5mm]
    \inode(PA-x)(1,1)
    \inode(PA-y)(1,2)
    \path
    (PA-x) edge[->,bend right] (PA-y) edge[->,bend left] (PA-y)
    ;
  \end{graph}
}
\newsavebox{\Bparedge}
\sbox \Bparedge{\Gparedge}

\section{Alternating  Graph Automata}\label{s:afga}

We now recall the definition of alternating graph automata of \cite{Bruggink-Huelsbusch-Koenig:12,dhm-icgt25a}.

\begin{defn}
  \label{d:afga}
  An \emph{alternating  graph automaton} over $\Voc$
  is given by $\aut A = (\Voc, Q, \Delta,\allowbreak q_0, Q_\forall)$ where
  \begin{itemize}
    \item $\Voc$ is a ranked set of edge labels,
    \item $Q$ is a ranked set of \emph{states},
    \item $\Delta \subseteq \bigcup_{m,n\in\Nat}\big(Q^{(m)} \times \GG_\Voc^{(m,n)} \times Q^{(n)}\big)$ is the
          set of \emph{transitions},
    \item $q_0\in Q$ is the \emph{initial state}, and
    \item $Q_\forall \subseteq Q$ is the subset of
          \emph{universal states}.
  \end{itemize}
  The set $Q \setminus Q_\forall$ of \emph{existential} states is denoted by $Q_\exists$.
  A \emph{permutation cycle} of $\aut A$ is a non-empty sequence
  $\delta_1\delta_2\cdots\delta_n\in\Delta^*$ where
  $\delta_i=(q_i,\Gamma_i,q_i')$, $\Gamma_i$ is a permutation graph,
  and $q_i'=q_{(i\bmod n)+1}$ for
  each $i\in[n]$.
\end{defn}

\begin{example}\label{x:auto}%
  \figref{f:tree-auto} shows an example of an alternating graph automaton. 
  In such illustrations of alternating graph automata, they
  are drawn as transition diagrams. The dangling incoming
  arrow indicates the initial state~$q_0$, and transitions
  $(q,\Gamma,q')$ are drawn as arrows from $q$ to $q'$ with the graph
  $\Gamma$ drawn on it. The interior of existential states is white with
  a small inscribed~$\exists$, and that of universal states is black
  with a small inscribed~$\forall$.
\end{example}
\begin{figure}[tb]
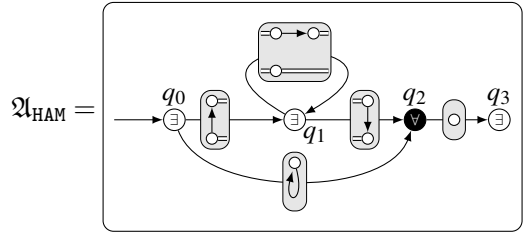

  \centering$
    \aut A_{\mathtt{HAM}} 
      =        
      \begin{automaton}[x=16mm,y=5mm]
        \node[state]           (q_0) at (0,0)   {\tiny$\exists$};
        \node[state]           (q_1) at (1,0)   {\tiny$\exists$};
        \node[state,universal] (q_2) at (2,0)   {\tiny$\forall$};
        \node[state]           (q_3) at (2.7,0) {\tiny$\exists$};
        \node[e]       at ($(q_0)+(0,0.6)$)     {\small$q_0$};
        \node[e]       at ($(q_1)+(0.17,-0.4)$) {\small$q_1$};
        \node[e]       at ($(q_2)+(0,0.6)$)     {\small$q_2$};
        \node[e]       at ($(q_3)+(0,0.6)$)     {\small$q_3$};
        \draw[->] (-0.5,0) -- (q_0);
        \path[->]
          (q_0) edge node[pos=0.3] {\UB{\BedgeU}} (q_1)
                edge[looseness=0.8,out=290,in=250] 
                     node {\UB{\Bloop}} (q_2)
          (q_1) edge[looseness=20,min distance=22mm,loop,out=150,in=30]
                     node {\UB{\BoutEN}} (q_1)
                edge node[pos=0.6] {\UB{\BedgeD}}  (q_2) 
          (q_2) edge node[pos=0.45] {\UB{\Bnode}}        (q_3) 
          ;
      \end{automaton}
  $
\caption{An alternating graph automaton accepting graphs containing a Hamiltonian cycle.}
\label{f:tree-auto}
\end{figure}

In~\cite{dhm-icgt25a}, the semantics of alternating graph automata is defined
by means of their \emph{configuration graphs}. In this paper, we will  
apply this approach not only to alternating graph automata, but also to define the semantics
of systems of graph formulas (in \sectref{s:graph formulas}). For this, we
now introduce \emph{evaluation graphs}, a slightly more abstract version of
configuration graphs.

Evaluation graphs have existential and universal nodes which, when instantiated to the case
of alternating graph automata, will correspond to configurations of the automaton in
existential and universal states, respectively. Iteratively, each node receives a truth
value as the disjunction or conjunction, respectively, of the truth values of its successors.
This process starts at the nodes which do not have outgoing edges, whose truth values are
either false (for existential nodes) or true (for universal ones). Since the graph may contain
cycles, the truth values at some nodes may remain indeterminate.

\begin{defn}[Evaluation Graph and Truth Assignment Evolution]\label{dfn:evolution}
  An \emph{evaluation graph} is an unlabeled graph in which every node is classified as
  either \emph{universal} or \emph{existential}.
  Given such an evaluation graph $\egraph$,
  we let~$\ass_\egraph$ denote the set of all partial functions $\alpha\colon\nd \egraph\to
    \{\true,\false\}$
  assigning truth values to (some of) the nodes in~$\egraph$. A truth assignment
  $\alpha\in\ass_\egraph$ \emph{evolves} (directly) into the truth assignment $\alpha'\in\ass_\egraph$
  given as follows, for all $c\in\nd \egraph$: if
  \begin{enumerate}[label=(\arabic*)]
    \item $c$ is universal and
          \begin{enumerate}[label=(\arabic{enumi}.\arabic*)]
            \item\label{case:U+} has only outgoing edges to nodes $c'\in\nd \egraph$ such that
                  $\alpha(c')=\true$ or
            \item\label{case:U-} has an outgoing edge to some $c'\in\nd \egraph$ such that
                  $\alpha(c')=\false$, or
          \end{enumerate}
    \item $c$ is existential and
          \begin{enumerate}[label=(\arabic{enumi}.\arabic*)]
            \item\label{case:E-} has only outgoing edges to nodes $c'\in\nd \egraph$ such that
                  $\alpha(c')=\false$ or
            \item\label{case:E+} has an outgoing edge to some $c'\in\nd \egraph$ such that
                  $\alpha(c')=\true$
          \end{enumerate}
  \end{enumerate}
  then
  \[
    \alpha'(c)=\left\{\begin{array}{ll}
      \true  & \text{in cases \ref{case:U+} and~\ref{case:E+}}  \\
      \false & \text{in cases \ref{case:U-} and~\ref{case:E-}.} \\
    \end{array}\right.
  \]
  If none of the conditions \ref{case:U+}--\ref{case:E+} is satisfied, then
  $\alpha'(c)=\alpha(c)$.

  The function that maps every truth assignment to the one it evolves
  into is denoted by $\evo_\egraph$, i.e., in the case above $\alpha'=\evo_\egraph(\alpha)$.
\end{defn}

Let $\alpha_\bot$ be the completely undefined truth assignment. It is shown in
\cite{dhm-icgt25a} (for configuration graphs, but the proof is the same) that the sequence
$(\evo_\egraph^i(\alpha_\bot))_{i\in\Nat}$ has a fixed point, in the following denoted by
$\evo_\egraph^*(\alpha_\bot)$. More precisely, for truth assignments
$\alpha$ and $\alpha'$, let $\alpha\sqsubseteq\alpha'$ if $\alpha'(c)=\alpha(c)$ for all
$c\in\nd \egraph$ such that $\alpha(c)$ is defined. Then the
following holds:

\begin{fact}[Part of {\cite[Lemma 1]{dhm-icgt25a}}]\label{fa:fp}
  For every evaluation graph $\egraph$, the fixed point $\alpha_E^*=\evo_\egraph^*(\alpha_\bot)$
  exists and it is the smallest fixed point of $\evo_\egraph$ with respect to $\sqsubseteq$.
\end{fact}

Intuitively, $\alpha_E^*=\evo_\egraph^*(\alpha_\bot)$ is the ``most undefined'' fixed point of
$\evo_\egraph$.
We now recall configuration graphs from \cite{dhm-icgt25a}, a particular type of evaluation
graphs.

\begin{defn}[Configuration Graph]\label{d:configs}
  A \emph{configuration} of an alternating graph automaton
  $\aut A = (\Voc, Q, \Delta,\allowbreak q_0,\allowbreak Q_\forall)$
  is a pair $(q,G)\in\bigcup_{m\in\Nat}\big(Q^{(m)}\times\GG_\Voc^m\big)$.
  It is \emph{universal} if $q$ is, and \emph{existential} otherwise.
  There is a \emph{(transition) step} $(q,G)\trans_\Delta(q',G')$ 
  if there is a transition $(q,\Gamma,q')\in\Delta$ such that
  (an isomorphic copy of) $\Gamma$ is a frontal subgraph of $G$ and
  $G'=G \off \Gamma$.

  The \emph{configuration graph} of $\aut A$ for an input graph
  $G_0\in\GG_\Voc^{\rank(q_0)}$ is the smallest graph $\cg_{\aut A}(G_0)$ over
  configurations containing the \emph{initial configuration} $(q_0,G_0)$ and, for
  each configuration $(q,G)$ that it contains and each step
  $(q,G)\trans_\Delta(q',G')$,
  the node $(q',G')$ and an edge from $(q,G)$ to $(q',G')$ representing this step.
\end{defn}

Obviously, configuration graphs are evaluation graphs. Hence, the semantics of
alternating graph automata can be defined using \factref{fa:fp}, as follows.

\begin{defn}[Languages Accepted by Alternating Graph Automata\label{dfn:semantics}]
  Let $\aut A = (\Voc, Q, \Delta,\allowbreak q_0, Q_\forall)$ be an alternating graph
  automaton.
  An input graph $G_0\in\GG_\Voc^{\rank(q_0)}$ is \emph{accepted} by $\aut A$
  if $\alpha_K^*(q_0,G_0)=\true$ and is \emph{rejected} by $\aut A$ if
  $\alpha_K^*(q_0,G_0)=\false$, where $K$ is the configuration graph of $\aut A$
  for $G$.

  The \emph{lower (upper) language accepted by $\aut A$} is the set
  $\underline L(\aut A)$ (the set $\overline L(\aut A)$, resp.) of all
  graphs
  $G_0\in\GG_\Voc^{\rank(q_0)}$ accepted (not rejected, resp.) by $\aut A$.
\end{defn}

Note that every state~$q$ without outgoing transitions can be considered as accepting or
rejecting, depending on whether it is universal or not.
This is because, for every configuration
$c\in\nd K$ which does not have outgoing edges in~$K$, by the definition
of $\evo_K$ it holds that
$\evo_K(\alpha)(c)=\true$ if $c$ is universal (by case \ref{case:U+}), and $\evo_K(\alpha)(c)=\false$ if~$c$ is
existential (by case \ref{case:E-}), for every assignment~$\alpha$. In particular, this is the case for $c=(q,G)$
if~$q$ has no outgoing transitions at all.

If $\alpha_K^*$ is a total function for all input graphs in
$\GG_\Voc^{\rank(q_0)}$, we have $\underline L(\aut A)=\overline L(\aut A)$. Then
we say that $L(\aut A)=\underline L(\aut A)$ is the \emph{language accepted by $\aut A$}.
As explained in~\cite{dhm-icgt25a}, a sufficient condition for this is that
$\aut A$ does not have permutation cycles. Then
each configuration graph is a directed acyclic graph (DAG) and
the construction of $\alpha_K^*$ as the fixed point of
$\evo_K^0(\alpha_\bot),\evo_K^1(\alpha_\bot),\evo_K^2(\alpha_\bot),\dots$ can be
replaced by a simple recursive evaluation of $\alpha_K^*(c)$ for all $c\in\nd K$.
%

As shown in \cite[Theorem~2]{dhm-icgt25a} alternating graph automata can accept
graph languages that are PSPACE-complete.

\begin{example}
  The automaton shown in \figref{f:tree-auto} detects whether a graph contains a
  Hamiltonian cycle or not, basically by checking whether there is a cycle in
  the graph such that there are no more nodes. It does this in the following
  way: Using the transitions from state~$q_0$ via~$q_1$ to~$q_2$, the automaton
  finds and follows a cycle of at least two binary edges. Using the transition
  from state~$q_0$ directly to~$q_2$ instead, it finds a single loop. In both cases the
  automaton ignores all other edges connected to nodes of the cycle
  (by cutting them off). The
  automaton accepts the graph if it finds no more nodes after reaching~$q_2$.
  Otherwise it reaches the rejecting state~$q_3$, so it rejects the graph.
\end{example}

\section{Graph Expressions and Systems of Graph Formulas}\label{s:graph formulas}

We now define graph expressions and, based on them, systems of graph formulas with
variables. Similarly to ordinary regular expressions for string languages, graph
expressions construct graph languages from singleton languages by means of union,
concatenation, and Kleene star as operators.


\begin{defn}[Graph Expressions]
  Let $i,j,k\in\Nat$.

  \begin{itemize}
    \item The (\emph{typed}) \emph{concatenation} of
          $\LL \subseteq \GG_\Voc^{(i,k)}$ and $\M \subseteq \GG_\Voc^{(k,j)}$ is
          $\LL \cdot \M = \{ G \cdot H \mid G \in \LL, H \in \M \}$.
    \item The (\emph{Kleene}) \emph{star} $\LL^\ast$
          of $\LL \subseteq \GG_\Voc^{(i,i)}$ is the smallest graph
          language containing $\Idg i$ and, for all graphs $G\in\LL$ and
          $H\in\LL^\ast$, the graph $G\cdot H\in\LL^\ast$.
  \end{itemize}

  The set $\Exp^{(i,j)}_\Voc$ of \emph{graph expressions} $\ex$ of type $(i,j)$ over
  $\Voc$ and the languages $L(\ex)$ they denote are defined inductively, as follows:
  \begin{enumerate}
    \item $\emptyset\in\Exp^{(i,j)}_\Voc$ with $L(\emptyset) = \emptyset$.
    \item If $G \in \GG_\Voc^{(i,j)}$, then
          $G \in \Exp^{(i,j)}_\Voc$ with $L(G) = \{ G \}$.
    \item If $\ex_1, \ex_2 \in \Exp^{(i,j)}_\Voc$, then 
          $\ex_1 \oplus\ex_2 \in \Exp^{(i,j)}_\Voc$ with
          $L(\ex_1 \oplus \ex_2) = L(\ex_1) \cup L(\ex_2)$.
    \item If $\ex_1 \in \Exp^{(i,j)}$ and $\ex_2 \in \Exp^{(j,k)}_\Voc$, then
          $\ex_1 \gcat \ex_2 \in \Exp^{(i,k)}_\Voc$ with
          $L(\ex_1 \gcat \ex_2) = L(\ex_1) \cdot L(\ex_2)$.
    \item If $\ex\in \Exp^{(i,i)}_\Voc$, then $\ex^\oast \in \Exp^{(i,i)}_\Voc$ with
          $L(\ex^\oast) = L(\ex)^\ast$.
  \end{enumerate}
\end{defn}

Note that graph concatenation is associative. For $\LL \subseteq \GG_\Voc^{(i,i)}$,
we may also abbreviate the $n$-fold iterated concatenation of $\LL$ with itself by
$\LL^n$, i.e., $\LL^0=\{\Idg i\}$ and $\LL^{n+1}=\LL\cdot\LL^n$ for $n\in\Nat$.

\def\GunedgeL{
  \begin{graph}[x=5mm,y=-5mm]
    \fnode(Gu-x)(1,1)
    \node[hyper] (Gu-L) at (2,1) {\scriptsize$\mathtt{L}$};
    \node[hyper,glass,minimum size=6mm]            at (2,1) {};
    \path
    (Gu-x) edge (Gu-L)
    ;
  \end{graph}
}
\begin{example}[Paths and Cycles]\label{x:cycle-ex}
  The following graph expressions specify a path from the front node
  to a node attached to a unary edge labeled \texttt{L}, and a cycle in
  a graph (with empty front interface), respectively.
  \begin{align*}
    \mathtt{Path}                                & =
    \raiseb\GoutE^\oast \!\gcat \;\raiseb\GunedgeL &
    \mathtt{Cycle}                               & =
    \raiseb\Gloop \oplus
    \raiseb\GedgeU \gcat\raisebox{0.6ex}{$\GoutEN^\oast$} \!\!\!\gcat \raiseb\GedgeD.
  \end{align*}
  In examples, $\oast$ binds stronger than $\gcat$, and $\gcat$ binds
  stronger than $\oplus$.
\end{example}

\begin{defn}[System of Graph Formulas with Variables]\label{d:eqsystem}%
  For $m\in\Nat$ and ranked sets $\Voc$ and $\Ident$ of edge labels and
  variables, respectively, the set $\RecFormula_{\Voc,\Ident}^m$ of
  \emph{graph formulas with variables} (\emph{formulas} for short) of type $m$ over $\Voc$ and
  $\Ident$ is defined inductively as follows:
  \begin{enumerate}[label=(\arabic*)]
    \item $\true, \false\in\RecFormula_{\Voc,\Ident}^m$;
    \item\label{c:id-case} $\Ident^{(m)}\subseteq\RecFormula_{\Voc,\Ident}^m$;
    \item if $\ex \in \EE_\Voc^{(m,i)}$ and $\fo \in \RecFormula_{\Voc,\Ident}^i$,
          then $\exists(\ex, \fo)\in \RecFormula_{\Voc,\Ident}^m$ and
          $\forall(\ex, \fo) \in \RecFormula_{\Voc,\Ident}^m$;
    \item if $\fo,\fo'\in\RecFormula_{\Voc,\Ident}^m$, then $\neg
            \fo\in\RecFormula_{\Voc,\Ident}^m$,
          $\fo\wedge\fo'\in\RecFormula_{\Voc,\Ident}^m$, and
          $\fo\vee\fo'\in\RecFormula_{\Voc,\Ident}^m$.
  \end{enumerate}
  A formula $\fo$ is called \emph{conjunctive} when $\fo = \true$, when it has
the form $\fo = \fo' \wedge \fo''$, or the form $\fo =
\forall(\ex, \fo')$. In all other cases, $\fo$ is called
\emph{disjunctive}.

  We let $\RecFormula_{\Voc,\Ident}=\bigcup_{m\in\Nat}\RecFormula^m_{\Voc,\Ident}$
  denote the set of all formulas over $\Voc$ and $\Ident$.

  A \emph{system of graph formulas with variables} (\emph{formula system} for short)
  is a function $F\colon\Ident\to\RecFormula_{\Voc,\Ident}$ such that
  $F(\id)\in\RecFormula^m_{\Voc,\Ident}$ for all $\id\in\Ident^{(m)}$.
\end{defn}

We will use the classification of formulas into conjunctive and disjunctive formulas 
later on when defining formula configurations.

In the following, unless explicitly stated otherwise, we assume fixed ranked
sets $\Voc$ and $\Ident$ of edge labels and variables, respectively.

Note that those formulas which can be built without using
case~\ref{c:id-case} in \defref{d:eqsystem}, are precisely the graph formulas
without variables as they are defined in~\cite{dhm-icgt25b}. In other words,
the latter set is a strict subset of the formulas considered here.

Before formalizing the semantics of formula systems, we define the special case
of \emph{acyclic} formula systems.

\begin{defn}[Acyclic formula system]\label{d:nonrec-eqsystem}%
  Given a \emph{formula system} $F\colon\Ident\to\RecFormula_{\Voc,\Ident}$, the
  \emph{dependency graph} $D_F$ of $F$ is an unlabeled graph defined as follows: Its node set is 
  $\nd D_F=\Ident$, and there is an edge from
  $u\in\Ident$ to $v\in\Ident$ whenever $v$ appears as a variable in $F(u)$. The
  formula system $F$ is \emph{acyclic} if $D_F$ is acyclic.
\end{defn}

We now define the semantics of formula systems. A
natural approach would be to define the semantics inductively based on the
structure of the formulas. However, this is challenging because a formula may
directly or indirectly depend on the variable to which it is assigned within
the system, potentially creating cyclic dependencies. For traditional
systems of language equations as introduced in the seminal paper by Mezei and
Wright \cite{Mezei-Wright:67} such dependencies do not pose problems as there
is a least fixed point, i.e., there is a smallest 
assignment of languages to the variables (with respect to set inclusion) such that the language equations are
fulfilled. This is due to the monotonicity of the involved operations. Here,
such monotonicity is lacking, essentially because formulas may contain
negations.

To deal with this situation, 
we adopt an approach
similar to that used for alternating graph automata. Specifically, we define the
semantics in terms of \emph{formula configuration graphs}, a specialized form of
evaluation graphs. This approach allows us to define the general semantics through
the fixed points of truth value assignments.

\begin{defn}[Formula Configuration Graph (FCG)]\label{d:rec-configs}%
  A \emph{formula configuration} is a triple $(\fo, G, \sign)$ where
  $\fo\in\RecFormula^m_{\Voc,\Ident}$, $G\in\GG_\Voc^m$, and
  $\sign\in\{\spos ,\sneg \}$.
  It is called \emph{universal} if $\sign = \spos $ and $\fo$ is conjunctive,
  or $\sign = \sneg $ and $\fo$ is disjunctive. Otherwise, it is said to be
  \emph{existential}.

  Given a formula system $F\colon\Ident\to\RecFormula_{\Voc,\Ident}$,
  there is a \emph{step} $(\fo,G,sign)\evalstep_F(\fo',G',\sign')$ from
  the formula configuration $(\fo,G,\sign)$ to $(\fo',G',\sign')$ in the
  following cases, for all $\otimes\in\{\vee,\wedge\}$, $\Quant\in\{\forall,\exists\}$,
  and $\id\in\Ident$:
  \begin{itemize}
    \item $(\neg\fo,G,\sign)\evalstep_F(\fo,G,\overline\sign)$ where
          $\overline\spos =\sneg $ and $\overline\sneg =\spos $,
    \item $(\fo_1\otimes\fo_2,G,\sign)\evalstep_F(\fo_1,G,\sign)$ and
          $(\fo_1\otimes\fo_2,G,\sign)\evalstep_F(\fo_2,G,\sign)$,
    \item $(\Quant(\ex,\fo),G,\sign)\evalstep_F(\fo,G\off P,\sign)$ for all
          frontal subgraphs $P\in
            L(\ex)$ of~$G$, and
    \item $(\id,G,\sign)\evalstep_F(F(\id),G,\sign)$.
  \end{itemize}
  Given a formula system $F\colon\Ident\to\RecFormula_{\Voc,\Ident}$ and
  a variable $\id\in\Ident$, the \emph{formula configuration
    graph} (\emph{FCG} for short) of $F$ at $\id$ for an input graph
  $G_0\in\GG_\Voc^m$ is the smallest graph $\fcg_{F,\id}(G_0)$ over formula
  configurations containing the \emph{initial configuration}
  $(\id,G_0,\spos)$ and, for each formula configuration $(\fo,G,\sign)$
  that it contains and each step
  $(\fo,G,\sign)\evalstep_F(\fo',G',\sign')$, the node
  $(\fo',G',\sign')$ and an edge from $(\fo,G,\sign)$ to
  $(\fo',G',\sign')$ representing this step.
\end{defn}

Note that the FCG is finite for any formula system~$F$ and variable~$\id$.

\begin{example}\label{x:FCG}%
  Let $\Voc=\{\invis\}$ with $\rank(\invis)=2$, and
  let $\Ident = \{u,v,w\}$, where each variable has rank~0. We define the
  formula system $H\colon\Ident \to \RecFormula_{\Voc,\Ident}$ by:
  \[
    H(u) = \overbrace{\exists\left(\raiseb\Gloop, v\right)}^{\mu} \vee 
           \overbrace{\exists\left(\raiseb\GedgeU, w\right)}^{\tau}, \qquad
    H(v) = \neg\overbrace{\exists\left(\raiseb\Gnode, \true\right)}^\psi, \qquad
    H(w) = \overbrace{\exists\left(\raiseb\GoutEN, w\right)}^{\phi} \vee 
           \overbrace{\exists\left(\raiseb\GedgeD, v\right)}^{\rho}.
  \]
  We name subformulas $\mu$, $\tau$, $\psi$, $\phi$, and $\rho$ for reference in \figref{f:fcg}.

  The intuitive semantics of $H$ is as follows: the variable $u$ defines graphs
  that either consist of a node with a loop followed by a subgraph described by
  $v$ (subformula~$\mu$), or a graph that starts with an edge and continues with
  a subgraph described by $w$ (subformula~$\tau$). The variable $w$ is satisfied
  by graphs that either start with an edge and are recursively followed by
  another $w$ graph (subformula~$\phi$), or graphs that connect the two front
  interface nodes with a single edge followed by a subgraph satisfying $v$
  (subformula~$\rho$). Finally, $v$ is satisfied if there are no nodes left in
  the graph. All in all, $u$ specifies graphs that contain a Hamiltonian cycle.

\def\Hin{
  \begin{graph}[x=3mm,y=-5mm]
    \node[e] (f-1) at (-0.7,0) {};
    \inode(1)(0,0)
    \inode(2)(0,1)
    \path[->] (1) edge[bend left] (2) edge[loop right] () (2) edge[bend left] (1)
    ;
  \end{graph}
}
\def\HinNode{
  \begin{graph}[x=3mm,y=-5mm]
    \node[e] (f-1) at (-0.7,0) {};
    \inode(2)(0,1)
  \end{graph}
}
%
\def\HinTwo#1#2{
  \begin{graph}[x=4mm,y=-5mm]
    \node[e] (f-1) at (-0.7,0) {};
    \node[e] (f-2) at (-0.7,1) {};
    \inode(1)(0,0)
    \inode(2)(0,1)
    \begin{pgfonlayer}{middleground}
      \path (f-1) edge[overlay,-,double distance=1pt,double,out=0,in=180] (#1);
      \path (f-2) edge[overlay,-,double distance=1pt,double,out=0,in=180] (#2);
    \end{pgfonlayer}
    \path[->] (1) edge[bend left] (2) edge[loop right] () (2) edge[bend left] (1)
    ;
  \end{graph}
}
\def\HinTwoMinus#1#2{
  \begin{graph}[x=3mm,y=-5mm]
    \node[e] (f-1) at (-0.7,0) {};
    \node[e] (f-2) at (-0.7,1) {};
    \inode(1)(0,0)
    \inode(2)(0,1)
    \begin{pgfonlayer}{middleground}
      \path (f-1) edge[overlay,-,double distance=1pt,double,out=0,in=180] (#1);
      \path (f-2) edge[overlay,-,double distance=1pt,double,out=0,in=180] (#2);
    \end{pgfonlayer}
    \path[->] (1) edge[bend left] (2) edge[loop right] () 
    ;
  \end{graph}
}
\def\HinTwoMinusO#1#2{
  \begin{graph}[x=4mm,y=-5mm]
    \node[e] (f-1) at (-0.7,0) {};
    \node[e] (f-2) at (-0.7,1) {};
    \inode(1)(0,0)
    \inode(2)(0,1)
    \begin{pgfonlayer}{middleground}
      \path (f-1) edge[overlay,-,double distance=1pt,double,out=0,in=180] (#1);
      \path (f-2) edge[overlay,-,double distance=1pt,double,out=0,in=180] (#2);
    \end{pgfonlayer}
    \path[->] (2) edge[bend left] (1) (1) edge[loop right] () 
    ;
  \end{graph}
}
  \begin{figure}[t]
    \begin{minipage}[b]{0.2\linewidth}
      \centering$G = \raiseb\Hin$
    \caption{A graph}
    \label{f:fcg-inp}
  \end{minipage}
  \hfill
  \begin{minipage}[b]{0.78\linewidth}
    \centering
    \begin{tikzpicture}[x=40mm,y=-13mm]
      \node[acc] (0) at (0,0) {$\CFGr{u}{\Hin}{\spos}$};
      \node[acc] (0-1) at ($(0)+(0,1)$) {$\CFGr{\mu\vee\tau}{\Hin}{\spos}$};
      \node[rej] (0-1-1) at ($(0-1)+(-0.8,1)$) {$\CFGr{\mu}{\Hin}{\spos}$}; 
      \node[rej] (0-1-1-1) at ($(0-1-1)+(0,1)$) {$\CFGr{v}{\HinNode}{\spos}$}; 
      \node[rej] (0-1-1-1-1) at ($(0-1-1-1)+(0,1)$) {$\CFGr{\neg\psi}{\HinNode}{\spos}$}; 
      \node[rej] (0-1-1-1-1-1) at ($(0-1-1-1-1)+(0,1)$) {$\CFGr{\psi}{\HinNode}{\sneg}$}; 
      \node[rej] (0-1-1-1-1-1-1) at ($(0-1-1-1-1-1)+(0,1)$) {$\CFGr{\mathit{true}}{\Gempty}{\sneg}$}; 
      \node[acc] (0-1-2) at ($(0-1)+(0.8,1)$) {$\CFGr{\tau}{\Hin}{\spos}$}; 
      \node[acc] (0-1-2-1) at ($(0-1-2)+(-0.6,1)$) {$\CFGr{w}{\HinTwoMinus 12}{\spos}$}; 
      \node[acc] (0-1-2-1-1) at ($(0-1-2-1)+(0,1)$) {$\CFGr{\phi\vee\rho}{\HinTwoMinus 12}{\spos}$}; 
      \node[rej] (0-1-2-1-1-1) at ($(0-1-2-1-1)+(-0.3,1)$) {$\CFGr{\phi}{\HinTwoMinus 12}{\spos}$}; 
      \node[acc] (0-1-2-1-1-2) at ($(0-1-2-1-1)+(0.3,1)$) {$\CFGr{\rho}{\HinTwoMinus 12}{\spos}$}; 
      \node[acc] (0-1-2-1-1-2-1) at ($(0-1-2-1-1-2)+(0.6,1)$) {$\CFGr{v}{\Gempty}{\spos}$}; 
      \node[acc] (0-1-2-1-1-2-1-1) at ($(0-1-2-1-1-2-1)+(0,1)$) {$\CFGr{\neg\psi}{\Gempty}{\spos}$}; 
      \node[acc] (0-1-2-1-1-2-1-1-1) at ($(0-1-2-1-1-2-1-1)+(0,1)$) {$\CFGr{\psi}{\Gempty}{\sneg}$}; 
      \node[acc] (0-1-2-2) at ($(0-1-2)+(0.6,1)$) {$\CFGr{w}{\HinTwoMinusO 21}{\spos}$}; 
      \node[acc] (0-1-2-2-1) at ($(0-1-2-2)+(0,1)$) {$\CFGr{\phi\vee\rho}{\HinTwoMinusO 21}{\spos}$}; 
      \node[rej] (0-1-2-2-1-1) at ($(0-1-2-2-1)+(-0.3,1)$) {$\CFGr{\phi}{\HinTwoMinusO 21}{\spos}$}; 
      \node[acc] (0-1-2-2-1-2) at ($(0-1-2-2-1)+(0.3,1)$) {$\CFGr{\rho}{\HinTwoMinusO 21}{\spos}$}; 
	  \node[e] at ($(0)+(0.4,0)$)     {\small$\exists, \true$};
      \node[e] at ($(0-1)+(0.52,0)$) {\small$\exists, \true$};
      \node[e] at ($(0-1-1)+(0,-0.5)$) {\small$\exists, \false$};
      \node[e] at ($(0-1-1-1)+(-0.2,-0.43)$) {\small$\exists, \false$};
      \node[e] at ($(0-1-1-1-1)+(-0.2,-0.43)$) {\small$\exists, \false$};
      \node[e] at ($(0-1-1-1-1-1)+(-0.2,-0.43)$) {\small$\forall, \false$};
      \node[e] at ($(0-1-1-1-1-1-1)+(-0.2,-0.5)$) {\small$\exists, \false$};
      \node[e] at ($(0-1-2)+(0,-0.5)$) {\small$\exists, \true$};
      \node[e] at ($(0-1-2-1)+(0,-0.5)$) {\small$\exists, \true$};
      \node[e] at ($(0-1-2-1-1)+(-0.3,-0.5)$) {\small$\exists, \true$};
      \node[e] at ($(0-1-2-1-1-1)+(-0.18,-0.49)$) {\small$\exists, \false$};
      \node[e] at ($(0-1-2-1-1-2)+(0.15,-0.47)$) {\small$\exists, \true$};
      \node[e] at ($(0-1-2-2)+(0,-0.5)$) {\small$\exists, \true$};
      \node[e] at ($(0-1-2-2-1)+(0.3,-0.47)$) {\small$\exists, \true$};
      \node[e] at ($(0-1-2-2-1-1)+(-0.15,-0.47)$) {\small$\exists, \false$};
      \node[e] at ($(0-1-2-2-1-2)+(0.18,-0.47)$) {\small$\exists, \true$};
      \node[e] at ($(0-1-2-1-1-2-1)+(0.35,0)$) {\small$\exists, \true$};
      \node[e] at ($(0-1-2-1-1-2-1-1)+(0.4,0)$) {\small$\exists, \true$};
      \node[e] at ($(0-1-2-1-1-2-1-1-1)+(0.35,0)$) {\small$\forall, \true$};

      \path
      (0) edge[->] (0-1)
      (0-1) edge[->] (0-1-1) edge[->] (0-1-2)
      (0-1-1) edge[->] (0-1-1-1) 
      (0-1-1-1) edge[->] (0-1-1-1-1) 
      (0-1-1-1-1) edge[->] (0-1-1-1-1-1)
      (0-1-1-1-1-1) edge[->] (0-1-1-1-1-1-1) 
      (0-1-2)  edge[->] (0-1-2-1)  edge[->] (0-1-2-2) 
      (0-1-2-1)  edge[->] (0-1-2-1-1) 
      (0-1-2-1-1)  edge[->] (0-1-2-1-1-1)  edge[->] (0-1-2-1-1-2) 
      (0-1-2-1-1-2) edge[->] (0-1-2-1-1-2-1)
      (0-1-2-1-1-2-1) edge[->] (0-1-2-1-1-2-1-1)
      (0-1-2-1-1-2-1-1) edge[->] (0-1-2-1-1-2-1-1-1)
      (0-1-2-2)  edge[->] (0-1-2-2-1) 
      (0-1-2-2-1)  edge[->] (0-1-2-2-1-1)  edge[->] (0-1-2-2-1-2) 
      (0-1-2-2-1-2)  edge[->] (0-1-2-1-1-2-1)  
      ;
    \end{tikzpicture}
    \vspace*{-7mm}
    \caption{The formula configuration graph $\fcg_{H,u}(G)$}
    \label{f:fcg}
  \end{minipage}
\end{figure}

  \figref{f:fcg} shows the formula configuration graph $\fcg_{H,u}(G)$ for the
  graph $G$ shown in \figref{f:fcg-inp}. The labels next to the configurations
  indicate whether they are universal ($\forall$) or existential ($\exists$).
  The indicated truth values will become relevant in \exref{x:FCG-eval}. Note
  that the formula configuration graph $\fcg_{H,u}(G)$ is acyclic, in contrast
  to $H$ itself, which is cyclic.
\end{example}

Since every FCG $K$ is an evaluation graph, we can once again use
\defref{dfn:evolution} to describe how truth values assigned to formula
configurations evolve from an initial assignment. By \factref{fa:fp}, the most
undefined fixed point, denoted by $\alpha_K^*$, exists. 
Following the approach in \defref{dfn:semantics}, we thus define:

\begin{defn}[Languages Accepted by Formula Systems\label{dfn:f-semantics}]
  Given a formula system $F\colon\Ident\to\RecFormula_{\Voc,\Ident}$
  and a variable $\id\in\Ident$, an input graph $G_0\in\GG_\Voc^m$
  is \emph{accepted} by $F$ at $\id$ if $\alpha_K^*(\id,G_0,\spos)=\true$ and
  is \emph{rejected} by $F$ at $\id$ if $\alpha_K^*(\id,G_0,\spos)=\false$, 
  where $K$ is the formula configuration graph of $F$ at $\id$.

  The \emph{lower language accepted by $F$ at $\id$} is the set $\underline
    L_\id(F)$ of all graphs $G_0\in\GG_\Voc^m$ accepted by $F$ at $\id$, and the
  \emph{upper language accepted by $F$ at $\id$} is the set $\overline
    L_\id(F)$ of all graphs $G_0\in\GG_\Voc^m$ not rejected by $F$ at $\id$.
\end{defn}

In the special case where $\alpha_K^*$ is a total function for all input graphs in
$\GG_\Voc^m$, we have $\underline L_\id(F)=\overline L_\id(F)$. Then we say that
$L_\id(F)=\underline L_\id(F)$ is the \emph{language accepted by $F$ at~$\id$}.

Note the difference between this definition and \defref{dfn:semantics}: While
an automaton has a particular state
designated as its initial state, no particular variable is singled out as
``initial'' in a formula system.
Therefore, there is not a single pair of upper and lower languages defined by
a formula system, but each variable is associated with its own pair of
languages. To put it in another way, specifying a pair of upper and lower languages
by a formula system requires us to additionally say which variable we are considering.

\begin{example}\label{x:FCG-eval}%
  Continuing \exref{x:FCG}, let $K=\fcg_{H,u}(G)$. The truth value assigned by $\alpha_K^*$ is specified
  for each configuration in \figref{f:fcg}. Since $K$ is acyclic, the
  function $\alpha_K^*$ is total, and $G\in\underline L_u(H)\cap\overline L_u(H)$.
\end{example}

Let us now demonstrate that the semantics of graph formulas without variables, 
as defined in~\cite{dhm-icgt25b}, coincides with the semantics of acyclic 
formula systems.

Every variable $\id\in\Ident$ occurring in an acyclic formula system $F$ can be
identified with an acyclic formula $F^*(\id)$ in the sense of~\cite{dhm-icgt25b},
namely the one obtained by taking $F(\id)$ and recursively replacing every occurrence
of a variable $\id'$ in it by $F^*(\id')$. Conversely, every formula $\fo$
in the sense of~\cite{dhm-icgt25b} can be written as the acyclic formula
system $F(x)=\varphi$ with only one variable~$x$.
In other words, acyclic formula systems do indeed coincide with the graph formulas
of~\cite{dhm-icgt25b}. The following lemma shows that the semantics of acyclic 
formula systems coincides with the inductive semantics definition of graph formulas 
in~\cite{dhm-icgt25b}:

\begin{lemma}\label{l:inductive-semantics}
  For each acyclic formula system $F\colon \Ident\to\RecFormula_{\Voc,\Ident}$ and
  each variable $\id\in\Ident$, it holds that
  $\underline{L}_\id(F) = \overline{L}_\id(F) = \{G \mid G \satisfies_F \id\}$ where
  the satisfaction relation ${\satisfies}_F \subseteq
    \GG_\Voc^m\times\RecFormula_{\Voc,\Ident}^m$ is defined by induction as
    follows. It is the smallest
  relation such that the following hold for every graph $G\in \GG_\Voc^m$ (where
  $\fo,\fo'\in\RecFormula_{\Voc,\Ident}$ and
  $\id\in\Ident$):
  \begin{enumerate}[label=(\arabic*)]
    \item\label{i:sat-true-false} $G \satisfies_F \true$.
    \item\label{i:sat-exists}  $G \satisfies_F \exists(\ex, \fo)$ if $G \off P
            \satisfies_F \fo$ for some frontal subgraph $P \in L(\ex)$
          of~$G$.
    \item\label{i:sat-forall} $G \satisfies_F \forall(\ex, \fo)$ if $G \off P
            \satisfies_F \fo$ for all frontal subgraphs $P \in L(\ex)$
          of~$G$.
    \item\label{i:sat-neg} $G \satisfies_F \neg \fo$ if $G \not\satisfies_F
            \fo$.
    \item\label{i:sat-and} $G \satisfies_F \fo\wedge\fo'$ if $G \satisfies_F
            \fo$ as well as $G \satisfies_F \fo'$.
    \item\label{i:sat-or} $G \satisfies_F \fo\vee\fo'$ if $G \satisfies_F
            \fo$ or $G \satisfies_F \fo'$ (or both).
    \item\label{i:sat-id} $G \satisfies_F \id$ if $G \satisfies_F
            F(\id)$.
  \end{enumerate}
\end{lemma}

The definition of the satisfaction relation in this lemma is a straightforward extension of 
\cite[Definition~7]{dhm-icgt25b}, which defines satisfaction of 
graph formulas without variables. \lemmaref{l:inductive-semantics} simply adds case~\ref{i:sat-id} and
defines the satisfaction of variables by the satisfaction of the assigned formula.

\begin{proof}
  Let $F\colon \Ident\to\RecFormula_{\Voc,\Ident}$ be an acyclic formula system,
  $\id\in\Ident$ a variable, and $G_0\in\GG^{\rank(x)}_\Voc$ a graph. Define
  $K=\fcg_{F,\id}(G_0)$ as the FCG of $F$ at $\id$ for $G_0$.

  We first show that $K$ is acyclic. Suppose, for contradiction, that $K$
  contains a cycle. By the definition of the step relation $\evalstep_F$ in
  \defref{d:rec-configs}, such a cycle must include formula configurations
  $(y,G,\sign)$ and $(F(y),G,\sign) \in \nd K$ such that
  \[
    (y,G,\sign) \evalstep_F (F(y),G,\sign) \evalstep_F^* (y,G,\sign)
  \]
  for some $y\in\Ident$. This implies that the dependency graph $D_F$ of $F$
  contains a cycle involving $y$, contradicting the assumption that $F$ is
  acyclic. This proves that $K$ is acyclic.

  For each formula configuration $c=(\fo,G,\sign)\in\nd K$, we prove by
  induction on the length of the longest path starting at $c$ that
  \begin{equation}
    \alpha_K^*(c) =
    \begin{cases}
      \true  & \text{if $\sign=\spos $ and $G\satisfies_F\fo$, or $\sign=\sneg $ and $G\not\satisfies_F\fo$,} \\
      \false & \text{otherwise.}
    \end{cases}
    \label{eq:compute-alpha}
  \end{equation}

  For the base case, consider a configuration $c$ with no outgoing edges. Then
  $\fo$ must be of the form $\true$, $\false$, $\exists(\ex,\fo')$, or
  $\forall(\ex,\fo')$. We analyze the case $\fo = \forall(\ex,\fo')$; the others
  follow similarly.

  Since $\fo = \forall(\ex,\fo')$, the configuration $c$ is universal if
  $\sign=\spos $ and existential if $\sign=\sneg $, by \defref{d:rec-configs}. The
  fact that $c$ has no outgoing edges means that $G$ has no frontal subgraph in
  $L(\ex)$, which implies $G\satisfies_F\fo$ by case~\ref{i:sat-forall}.
  Furthermore, by \defref{dfn:evolution} the absence of edges leaving $c$ means that $\alpha_K^*(c) =
  \true$ if $\sign=\spos $ and $\alpha_K^*(c) = \false$ if $\sign=\sneg $. Hence,
  equation \eqref{eq:compute-alpha} holds for $c$ in both cases.

  For the inductive step, assume that $c$ has at least one successor
  configuration $c'$ with $c \evalstep_F c'$, and that \eqref{eq:compute-alpha}
  holds for all such $c'$. By \defref{d:rec-configs}, $\fo$ must be one of
  $\fo\in\Ident$, $\exists(\ex,\fo')$, $\forall(\ex,\fo')$, $\neg\fo'$, $\fo'
    \wedge \fo''$, and $\fo' \vee \fo''$. In each case, either $\sign=\spos $ or
  $\sign=\sneg $, and either $G\satisfies_F\fo$ or $G\not\satisfies_F\fo$. We analyze
  two representative cases: $\fo=\forall(\ex,\fo')$ and $\fo=\neg\fo'$, both
  with $\sign=\spos $ and $G\satisfies_F\fo$; the remaining cases follow similarly.

  \begin{itemize}
    \item Case $\fo = \neg\fo'$, $\sign = \spos $, and $G \satisfies_F \fo$.

          Since $\fo = \neg\fo'$, the only successor configuration is $c' =
            (\fo',G,\sneg)$. Since $G\satisfies_F\fo$, we have
          $G\not\satisfies_F\fo'$, so the induction hypothesis gives
          $\alpha_K^*(c') = \true$. By \defref{dfn:evolution}, this implies
          $\alpha_K^*(c) = \true$, as required by equation~\eqref{eq:compute-alpha}.

    \item Case $\fo = \forall(\ex,\fo')$, $\sign = \spos $, and $G \satisfies_F \fo$.

          By assumption, $G\satisfies_F\fo$. This
          can only be the case if $G \off P \satisfies_F \fo'$ holds for all
          frontal subgraphs $P \in L(\ex)$. By \defref{d:rec-configs}, each
          successor configuration of~$c$ is of the form $c' = (\fo', G \off P,
          \spos)$. By the induction hypothesis, $\alpha_K^*(c') = \true$ for all such
          $c'$, and by \defref{dfn:evolution} this implies $\alpha_K^*(c) =
          \true$, as required by equation~\eqref{eq:compute-alpha}, because $c$ is
          universal.
  \end{itemize}

  Since equation~\eqref{eq:compute-alpha} holds for all configurations in $K$, it
  particularly holds for the initial configuration $(\id, G_0, \spos)$ of $K$.
  Consequently,
  \[
    \alpha_K^*(\id,G_0,\spos) =
    \begin{cases}
      \true  & \text{if $G_0\satisfies_F\id$,}     \\
      \false & \text{if $G_0\not\satisfies_F\id$.}
    \end{cases}
  \]
  Thus, $G_0 \satisfies_F\id$ if and only if $G_0 \in L_\id(F)$, completing
  the proof.
\end{proof}

Note that, as one would expect, by \defref{dfn:f-semantics} the semantics of 
disjunction and conjunction of formulas is both commutative and associative.
This follows directly from the fact that FCGs contain no information about the order
of arguments of disjunctions and conjunctions, i.e., interchanging them results in
the same FCG. Thus, given a (finite) set $\mathit{FO}=\{\fo_1,\dots,\fo_n\}$ of
formulas, we can express their disjunction and conjunction using the standard notations
\[
  \bigvee_{\fo\in\mathit{FO}} \fo = \fo_1 \vee \fo_2 \vee \dots \vee \fo_n,
  \quad \text{and} \quad
  \bigwedge_{\fo\in\mathit{FO}} \fo = \fo_1 \wedge \fo_2 \wedge \dots \wedge \fo_n
\]
as shorthands for sequences of binary disjunctions and conjunctions, respectively.

In the next section we will show that formula systems have the same expressive
power as alternating graph automata. To facilitate that proof, but also because it
is of independent interest, we now develop a normal form of formula systems which
we call \emph{shallow normal form}. The
result states that every formula system~$F$ can be
brought into a structurally simpler normal form $F'$ over a larger set of variables
$\Ident'$ such that, for all $\id\in\Ident'$, every proper sub-formula of $F(\id)$
is a variable. 

\begin{theorem}[Shallow Normal Form]\label{th:nf}%
  For each formula system $F\colon \Ident\to\RecFormula_{\Voc,\Ident}$ there is a 
  formula system $F'\colon\Ident'\to\RecFormula_{\Voc,\Ident'}$ with $\Ident'\supseteq\Ident$
  such that 
  \begin{enumerate}
  \item $\underline L_\id(F')=\underline L_\id(F)$ and 
        $\overline L_\id(F')=\overline L_\id(F)$ for all $\id\in\Ident$, and
  \item for every $\id\in\Ident'$, $F'(\id)$ has one of the following forms, where
        $G\in\GG_\Voc$ and $\id',\id''\in\Ident'$:
        \[
          \true,\ \false,\ \id'\vee\id'',\  \id'\wedge\id'',\ \exists(G,\id') \text{, and }\forall(G,\id').
        \]
  \end{enumerate}
\end{theorem}

The proof of the existence of the shallow normal form makes use of an auxiliary lemma.

\begin{lemma}\label{lem:non-permuting}
  Call a graph expression \emph{non-permuting} if it has no sub-expression of the form
  $\ex^\oast$ such that $L(\ex)$ contains a permutation graph. Then every graph expression has
  an equivalent non-permuting graph expression.
\end{lemma}

\begin{proof}
Denote by $P^{(i)}$ the set of all permutation graphs of type $(i,i)$.
We show the following statement by structural induction:
for every graph expression $\ex\in\Exp^{(i,j)}_\Voc$,
there is a non-permuting graph expression $\ex_0\in\Exp^{(i,j)}_\Voc$ such that $L(\ex_0)=L(\ex)\setminus
P^{(i)}$. This proves the lemma because the set $P$ of permutation graphs in $L(\ex)$ is finite,
say $P=\{\pi_1,\dots,\pi_k\}$, and thus $\ex_0\oplus\pi_1\oplus\cdots\oplus\pi_k$ is a
non-permuting graph expression equivalent to $\ex$.

The induction is straightforward: if $\ex=G\in\GG_\Voc^{(i,j)}$ then
$\ex_0=G$ unless~$G$ is a permutation graph, in which case $\ex_0=\emptyset$. If
$\ex=\ex'\oplus\ex''$ then $\ex_0=\ex'_0\oplus\ex''_0$, where $\ex'_0$ and $\ex''_0$
are obtained from $\ex'$ and $\ex''$ using the induction hypothesis.
If $\ex=\ex'\odot\ex''$ then $\ex_0=(\ex'_0\odot\ex'')\oplus(\ex'\odot\ex''_0)$.
Finally, if $\ex=(\ex')^\oast$ 
and $\{\pi_1,\dots,\pi_k\}$ is the set of permutation graphs in $L(ex)$
then $$\ex_0=\Pi\odot\ex'_0\odot(\Pi\odot\ex'_0)^\oast\odot\Pi$$
where $\Pi=\pi_1\oplus\cdots\oplus\pi_k$.

Correctness follows immediately from the relevant definitions, especially the fact
that the concatenation of graphs $G$ and $G'$ is a permutation graph if and only if both 
$G$ and $G'$ are.
\end{proof}

\begin{proof}[Proof of \thmref{th:nf}]
The construction of the shallow normal form $F'$ of $F$ proceeds in two steps. First we remove negations by
``pushing negations down'' through the formulas. Second, we repeatedly decompose formulas into smaller ones.

For the first step, let $\widebar\id$ be a fresh copy of each variable $\id\in\Ident$,
and extend $F$ by adding the equation $F(\widebar\id)=\neg F(\id)$ for every $\id\in\Ident$.
Clearly, this does not affect any of the languages accepted or rejected at the original
variables $\id$ because no $F(\id)$, $\id\in\Ident$, contains one of the new variables. Consequently, the language accepted (rejected) at $\widebar\id$ is the language rejected (accepted, respectively) at $\id$, for every $\id\in\Ident$.
In the following, we let $\widebar\Ident=\{\widebar\id\mid\id\in\Ident\}$ and
$\widebar{\widebar\id}=\id$ for all $\id\in\Ident$.

Now, as long as the formula system obtained in this way still contains negations, pick
any sub-formula of the form $\neg\fo$ and replace it by
\[
\left\{\begin{array}{ll@{}}
\fo' & \text{if $\fo=\neg\fo'$}\\
\neg\fo'\opname{\widebar\otimes}\neg\fo'' & \text{if $\fo=\fo'\otimes\fo''$ for some $\otimes\in\{{\vee},{\wedge}\}$, where $\widebar\vee=\wedge$ and $\widebar\wedge=\vee$}\\
\widebar\Quant(\ex,\neg\fo') & \text{if $\fo=\Quant(\ex,\fo')$ for some $\Quant\in\{{\exists},{\forall}\}$, where $\widebar\exists=\forall$ and $\widebar\forall=\exists$}\\
\widebar\id & \text{if $\fo=\id\in\Ident\cup\widebar\Ident$.}\\
\widebar\fo & \text{if  $\fo\in\{\true,\false\}$, where $\widebar\true=\false$ and $\widebar\false=\true$}\\
\end{array}\right.
\]
Obviously, this transformation procedure terminates after a finite number of steps and results in a formula system without occurrences of $\neg$.

Using the standard rules of predicate logic (in particular deMorgan's rules) and the relation between the languages accepted and rejected at $\id$ and $\widebar\id$, it can be verified in a straightforward manner that each of the transformation steps preserves those languages, and thus by induction the languages accepted and rejected by the resulting formula system at each $\id\in\Ident$ are preserved as well.
In the following, second step of the transformation to normal form, we can thus assume without loss of generality that the original system $F$ does not contain negations.

Next, we transform the formula system iteratively until it is in shallow normal form. In 
doing so, we may keep equations of the form $F'(\id)=\id'$, because such an equation can
obviously be replaced by $F'(\id)=\id'\wedge\id'$ without affecting the most undefined
fixed point $\alpha^*$.

If $F$ is is not yet in shallow normal form, pick an $\id\in\Ident$
such that $F(\id)$ is not as required. We decompose $F(\id)$ by introducing one or two fresh
variables and replacing the equation for $F(\id)$ by two or more equations. For all
$\id'\in\Ident\setminus\{\id\}$ we let $F'(\id')=F(\id')$. We distinguish the relevant cases,
where $\id_1$ and $\id_2$ are fresh variables of appropriate types.

If $F(\id)=\fo_1\otimes\fo_2$ with ${\otimes}\in\{\vee,\wedge\}$ and
$\{\fo_1,\fo_2\}\nsubseteq\Ident$, we let
$F'(\id)=\id_1\otimes\id_2$ and $F'(\id_i)=\fo_i$ for $i\in[2]$.

If $F(\id)=\Quant(G,\fo)$ with $\Quant\in\{\exists,\forall\}$, $G\in\GG_\Voc$, and $\fo\notin\Ident$, let 
$F'(\id)=\Quant(G,\id_1)$ and $F'(\id_1)=\fo$.

Finally, let $F(\id)=\Quant(\ex,\fo)$ where $\Quant\in\{\exists,\forall\}$ and $\ex\notin\GG_\Voc$. By
\lemmaref{lem:non-permuting}, we may assume
without loss of generality that $\ex$ is non-permuting. This case
has a number of sub-cases depending on the structure of $\ex$, as follows.

If $\ex=\ex_1\oplus\ex_2$, we define
\[
F'(\id)=\left\{\begin{array}{ll@{}}
\id_1\vee\id_2 & \text{if $\Quant={\exists}$}\\
\id_1\wedge\id_2 & \text{if $\Quant={\forall}$}\\
\end{array}\right.
\text{\quad and\quad $F'(\id_i)=\Quant(\ex_i,\fo)$ for $i\in[2]$.}
\]

If $\ex=\ex_1\odot\ex_2$, we define $F'(\id)=\Quant(\ex_1,\id_1)$ and $F'(\id_1)=\Quant(\ex_2,\fo)$, which
is semantically equivalent by \factref{f:off-gcat}.

If $\ex=\ex_1^\oast$, we define
\[
F'(\id)=\left\{\begin{array}{ll}
\id_1\vee\id_2 & \text{if $\Quant={\exists}$}\\
\id_1\wedge\id_2 & \text{if $\Quant={\forall}$}\\
\end{array}\right.
\text{\quad $F'(\id_1)=\fo$\quad and\quad $F'(\id_2)=\Quant(\ex_1,\id)$}
\]
again making use of \factref{f:off-gcat}.

This construction can only be repeated a finite number of times (bounded from above by the
sum of the sizes of the formulas $F(\id)$) because in each step the sum of the sizes of those
formulas which violate normal form is reduced by one.

For the correctness proof of the construction, we need to show that every step of the
transformation preserves the languages at each of the original variables. Rather than showing
this in detail for each case, we look at one example case because
the rest is similar. We consider the least obvious one of the cases, namely $F(\id)=\Quant(\ex,\fo)$ 
where $\ex=\ex_1^\oast$. Assume that ${\Quant}={\forall}$ (the case ${\Quant}={\exists}$ is dual and thus
follows by the same arguments) and consider FCGs~$\fcg_{F,\id_0}(G_0)$ and $\fcg_{F',\id_0}(G_0)$,
respectively, for some input graph~$G_0$.

\newcommand{\frontal}{\mathit{FG}}%
\newcommand{\nxt}{\mathit{NEXT}}%
In the rest of the proof, we denote the set of all frontal subgraphs of a graph~$G$ by
$\frontal(G)$.

Now, consider a configuration $\kappa=(\id,G,\sign)$ in $\fcg_{F,\id_0}(G_0)$. By the definition of
$\fcg_{F,\id_0}(G_0)$, $\kappa$ has the unique successor $(\forall(\ex,\fo),G,\sign)$. 
In turn, the successors of that configuration are all $(\fo,G',\sign)$ such that $G'\in\nxt$, where
\[
\nxt=\{G\off\Gamma\mid\Gamma\in\frontal(G)\cap L(\ex)\}.
\]
Since $\ex=\ex_1^\oast$ and $\nxt$ is finite, there is some $k\in\Nat$ such that $\nxt=
\nxt_0\cup\cdots\cup\nxt_k$ where $\nxt_0=\{G\}$ and $\nxt_{i+1}=\{G'\off\Gamma\mid G'\in\nxt_i,\ 
\Gamma\in\frontal(G')\cap L(\ex_1)\}$ for all
$i\in\Nat$. Note that, since $L(\ex_1)$
does not contain permutation graphs, each $G'\off\Gamma$ in this construction is strictly smaller
than~$G'$.

$\fcg_{F',\id_0}(G_0)$ differs from $\fcg_{F,\id_0}(G_0)$ in that the edges from~$\kappa$ to $(\fo,G',\sign)$
($G'\in\nxt$) are replaced by a more complex structure. Since
\[
F'(\id)=\left\{\begin{array}{ll}
\id_1\vee\id_2 & \text{if $\Quant={\exists}$}\\
\id_1\wedge\id_2 & \text{if $\Quant={\forall}$}\\
\end{array}\right.
\text{\quad $F'(\id_1)=\fo$\quad and\quad $F'(\id_2)=\Quant(\ex_1,\id)$,}
\]
this structure is defined recursively, as follows: For every graph~$G'\in\nxt_i$ such that
$\fcg_{F',\id_0}(G_0)$ contains the configuration $\kappa'=(\id,G',\sign)$ (the initial case being
$\kappa'=\kappa$), there is an edge from~$\kappa'$ to $(\id_1\wedge\id_2,G',\sign)$ and there are edges from
the latter to $(\id_1,G',\sign)$ and $(\id_2,G',\sign)$. These have the unique successors
$(\fo,G',\sign)$ and $(\forall(\ex_1,\id),G',\sign)$, respectively. Of these, the former is
in~$\fcg_{F',\id_0}(G_0)$ (the recursive construction stops) whereas the latter has all $(\id,G'',\sign)$ as
successors such that $G''\in\nxt_{i+1}$. The fact that~$G''$ is strictly smaller than~$G'$ ensures
that the subgraph of $\fcg_{F',\id_0}(G_0)$ created in this way is a DAG whose leaves are precisely
the elements of $\nxt$, i.e., the successors of $(\fo,G',\sign)$ in $\fcg_{F,\id_0}(G_0)$.
It follows that $\alpha_{\fcg_{F',\id_0}(G_0)}^*(\kappa')=\alpha_{\fcg_{F,\id_0}(G_0)}^*(\kappa)$, as
claimed.
\end{proof}

\section{Expressive Power of Formula Systems}\label{s:power}

In this section, we investigate the expressive power of formula systems
and establish their equivalence to alternating graph automata. Specifically, we
show that every alternating graph automaton can be transformed into a formula
system that yields the same upper and lower languages. Conversely, we
construct an alternating graph automaton for any given formula system $F$ 
in shallow normal form and any
$\id\in\Ident$, such that the languages accepted and rejected by $F$ at $\id$
are precisely those accepted and rejected, respectively, by the automaton.

\begin{lemma}\label{l:RGF-for-auto}%
  Let $\aut A = (\Voc, Q, \Delta, q_0, Q_\forall)$ be an alternating graph
  automaton. Define the function $F\colon Q \to \RecFormula_{\Voc,Q}$ by
  \begin{equation}
    F(q) =
    \begin{cases}
      \displaystyle\bigwedge_{(q,\Gamma,q')\in\Delta} \!\!\!\!\!\!\forall(\Gamma, q') & \text{if } q \in Q_\forall, \\
      \displaystyle\bigvee_{(q,\Gamma,q')\in\Delta} \!\!\!\!\!\!\exists(\Gamma, q')   & \text{otherwise}
    \end{cases}
  \label{eq:RGF-from-auto}
  \end{equation}
  for each state $q \in Q$.
  Then $F$ is a formula system and the following equivalences hold:
  \[
    \underline{L}(\aut A) = \underline{L}_{q_0}(F) \qquad \text{and}
    \qquad \overline{L}(\aut A) = \overline{L}_{q_0}(F).
  \]
\end{lemma}

\begin{proof}
  Let $\aut A$ and $F$ be as described in
  the lemma, and consider a graph $G_0 \in \GG^{\rank(q_0)}_\Voc$. Define $K = \cg_{\aut
      A}(G_0)$ and $K' = \fcg_{F,q_0}(G_0)$ to be the configuration graph of $\aut A$
  and the formula configuration graph of $F$ at $q_0$, respectively, both
  constructed for the input graph $G_0$.
  Clearly, $F$ is a formula system since $F(q)\in\RecFormula^m_{\Voc,Q}$ for every $q\in Q^{(m)}$.

  Next, we show that every configuration $(q,G) \in \nd K$ also appears in $\nd K'$
  as $(q,G,\spos)$. This is proven by induction on the length of the
  shortest path from the initial configuration $(q_0,G_0)$ to $(q,G)$ in $K$.
  For the base case, the claim holds trivially since $(q_0,G_0,\spos)$ is
  the initial configuration of $K'$.
  For the inductive step, assume that $(q,G)$ and $(q',G')$ are configurations
  in $K$ such that $(q,G) \trans_\Delta (q',G')$, and assume, as the induction
  hypothesis, that $(q,G,\spos) \in \nd K'$. Further, let $q$ have exactly $n \in
    \Nat$ outgoing transitions $(q, \Gamma_1, q_1), \dots, (q, \Gamma_n, q_n) \in
    \Delta$. Consequently, there exists an index $i \in [n]$ and a frontal
  subgraph $P$ of $G$ such that $q' = q_i$, $\Gamma_i$ is isomorphic to~$P$, and
  $G' = G \off P$.
  We now consider the case where $q$ is universal; the existential case follows
  analogously. Since $F(q) = \bigwedge_{i\in[n]} \forall(\Gamma_i, q_i)$, the
  configuration $(q,G,\spos)$ in $K'$ has $(F(q),G,\spos)$ as its only successor. This
  configuration, in turn, has successor configurations $(\forall(\Gamma_j, q_j),
    G,\spos)$ in $K'$ (along with possibly some intermediate configurations) for all
  $j \in [n]$. An example for $n=2$ is shown in \figref{f:fcg-succ}.
  As a result, $(\forall(\Gamma_i, q_i), G,\spos) \evalstep_F (q', G',\spos)$, which
  implies that $(q', G',\spos) \in \nd K'$, completing the induction.
  \begin{figure}[tb]
    \centering
    \begin{tikzpicture}[x=20mm,y=-10mm]
      \node[config] (0) at (0,0) {$\CFG{q}{G}{\spos}$};
      \node[config] (0-1) at ($(0)+(0,1)$)
      {$\CFG{\forall\left(\Gamma_1,q_1\right) \wedge \forall\left(\Gamma_2,q_2\right)}{G}{\spos}$};
      \node[config] (0-1-1) at ($(0-1)+(-0.8,1)$) {$\CFG{\forall\left(\Gamma_1,q_1\right)}{G}{\spos}$}; 
      \node[config] (0-1-2) at ($(0-1)+(0.8,1)$) {$\CFG{\forall\left(\Gamma_2,q_2\right)}{G}{\spos}$}; 
            
      \path
      (0) edge[->] (0-1)
      (0-1) edge[->] (0-1-1) edge[->] (0-1-2)
      ;
    \end{tikzpicture}
    \caption{Formula configuration $(q,G,\spos)\in\nd K'$ and (some of) its successors in $K'$.}
    \label{f:fcg-succ}
  \end{figure}
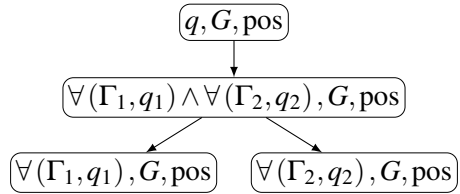

  By a very similar argument, we can show that every formula configuration
  $(q,G,\spos)\in\nd K'$ also appears in $\nd K$ as $(q,G)$.

  Next, we examine the most undefined fixed points $\alpha_K^*$ and
  $\alpha_{K'}^*$ for $K$ and $K'$, respectively. First, observe that for every
  formula configuration $(q,G,\spos) \in \nd K'$, we have
  $
    \alpha_{K'}^*(q,G,\spos) = \alpha_{K'}^*(F(q),G,\spos),
  $
  since $(F(q),G,\spos)$ is the only successor of $(q,G,\spos)$ in $K'$.\footnote{For
    simplicity, we also write $\alpha_K^*(c) = \alpha_{K'}^*(c')$ if both
    $\alpha_K^*(c)$ and $\alpha_{K'}^*(c')$ are undefined.}

  Now consider any formula configuration $(F(q),G,\spos) \in \nd K'$. We focus on
  the case where $q$ is universal, i.e., $F(q) = \bigwedge_{i\in[n]}
    \forall(\Gamma_i, q_i)$ for some $n \in \Nat$; again, the existential case is
  analogous. The configuration $(F(q),G,\spos)$ is the root of a subtree in $K'$
  whose leaves are the formula configurations $(\forall(\Gamma_i, q_i), G,\spos)$
  for each $i \in [n]$. \figref{f:fcg-succ} illustrates this structure for
  $n=2$. All nodes in this tree are universal.

  As shown earlier, every step $(q,G) \trans_\Delta (q',G')$ in $K$ corresponds
  to a step $(\forall(\Gamma_i, q_i), G,\spos) \evalstep_F \allowbreak(q', G',\spos)$
  in $K'$, and the reverse correspondence also holds. This establishes that the
  evolution of truth values in $K$ and $K'$ results in fixed points $\alpha_K^*$
  and $\alpha_{K'}^*$ such that
  $
    \alpha_K^*(q,G) = \alpha_{K'}^*(q,G,\spos)
  $
  for all configurations $(q,G) \in \nd K$, and in particular for the initial
  configurations, yielding
  \[
    \alpha_K^*(q_0,G_0) = \alpha_{K'}^*(q_0,G_0,\spos),
  \]
  which completes the proof.
\end{proof}

\begin{example}
  Consider the alternating graph automaton in \figref{f:tree-auto} with the set
  $Q=\{q_0,q_1,q_2,q_3\}$ of states. Its equivalent formula system  $F\colon Q
  \to \RecFormula_{\Voc,Q}$ is defined by:
  $$
    q_0 \mapsto \exists\left(\raiseb\Gloop, q_2\right) \vee 
           \exists\left(\raiseb\GedgeU, q_1\right) \quad
    q_1 \mapsto \exists\left(\raiseb\GoutEN, q_1\right) \vee 
           \exists\left(\raiseb\GedgeD, q_2\right) \quad
    q_2 \mapsto \forall\left(\raiseb\Gnode, q_3\right) \quad
    q_3 \mapsto \false
  $$
  Note that $F(q_2)$ and $F(q_3)$ can be combined to $F(q_2) =
  \forall\left(\raiseb\Gnode, \false\right)$, which is equivalent to $H(v)$ in
  \exref{x:FCG}. This shows that \exref{x:FCG} indeed specifies the language of
  all graphs with a Hamiltonian cycle.
  \qed
\end{example}

To prove the converse, namely that alternating graph automata are at least as
expressive as formula systems, we show that every formula system can be
transformed into an equivalent alternating graph automaton. Since every formula system can be
transformed into shallow normal form, it is sufficient to
consider only formula systems in this normal form.

\begin{lemma}\label{l:RGF-to-auto}%
  Let $F\colon \Ident\to\RecFormula_{\Voc,\Ident}$ be a formula system in shallow
  normal form, and let $\id_0\in\Ident$ be any variable.
  Define the alternating graph automaton  
  $\aut A = (\Voc, Q, \Delta, q_0, Q_\forall)$ with
  $Q = \Ident$,
	  $q_0 = \id_0$,
	  $Q_\forall = \{\id\in\Ident\mid F(\id)\text { is conjunctive}\}$,
  and $\Delta$ being the smallest set such that:
  \begin{itemize}
    \item $(\id,\Idg m,\id'),(\id,\Idg m,\id'')\in\Delta$ for each
          $\id\in\Ident^{(m)}$ such that $F(\id)$ is of the form $\id'\vee\id''$
          or $\id'\wedge\id''$.
    \item $(x,\Gamma,x')\in\Delta$ for each $\id\in\Ident$ such that $F(\id)$ is
          of the form $\forall(\Gamma,\id')$ or  $\exists(\Gamma,\id')$.
   \end{itemize}
   Then, the following equalities hold:
   $$
    \underline{L}(\aut A) = \underline{L}_{\id_0}(F) \quad
    \text{and} \quad \overline{L}(\aut A) = \overline{L}_{\id_0}(F).
   $$ 
\end{lemma}

\begin{proof}
   Let $F$, $\id_0$, and $\aut A$ be as described in the lemma. Now construct
   the formula system $F'\colon\Ident\to\RecFormula_{\Voc,\Ident}$ from~$\aut A$
   as described in \lemmaref{l:RGF-for-auto}. Note that both $F$ and $F'$ are
   defined for the same set $\Ident$ of variables. The following equality
   follows immediately from the construction of~$\aut A$ and the definition
   of~$F'$ by~\eqref{eq:RGF-from-auto}, for each $\id\in\Ident$:
   \begin{equation}
      F'(\id) = 
      \left\{
      \begin{array}{ll}
      \exists(\Idg m,\id')\vee\exists(\Idg m,\id'') & \text{if } F(\id)=\id'\vee\id''\\
      \forall(\Idg m,\id')\wedge\forall(\Idg m,\id'') & \text{if } F(\id)=\id'\wedge\id''\\
      F(x)&\text{otherwise}
      \end{array}
      \right.
      \qquad\text{ where }m=\rank(\id).
      \label{eq:eqivalent}
   \end{equation}
   Note that each graph formula $\fo\in\RecFormula_{\Voc,\Ident}^m$ is
   semantically equivalent to both $\exists(\Idg m,\fo)$ and $\forall(\Idg
   m,\fo)$. This is easy to see when considering the formula configurations
   \[\left(\exists(\Idg
   m,\fo),G,\sign\right)\qquad\text{and}\qquad\left(\forall(\Idg
   m,\fo),G,\sign\right)\] for any graph $G\in\GG^{(m)}_\Voc$ and
   $\sign\in\{\spos,\sneg\}$. For both, $(\fo,G,\sign)$ is their only successor
   configuration in any FCG $E$, and hence \[\alpha_E^*\left(\exists(\Idg
   m,\fo),G,\sign\right)=\alpha_E^*\left(\forall(\Idg
   m,\fo),G,\sign\right)=\alpha_E^*(\fo,G,\sign),\] showing the equivalence of
   $\fo$, $\exists(\Idg m,\fo)$ and $\forall(\Idg m,\fo)$. Consequently, the
   equality \eqref{eq:eqivalent} can be reduced to $F'=F$, which completes the
   proof.
\end{proof}

\lemmaref{l:RGF-for-auto} and \lemmaref{l:RGF-to-auto} establish a bidirectional
correspondence between alternating graph automata and formula systems. In
particular, they show that any alternating graph automaton can be transformed
into an equivalent formula system, and conversely, any formula system, if it is
in shallow normal form, can be translated into an equivalent
alternating graph automaton. Since every formula system can be transformed into
this normal form, this leads to the main result of this paper:

\begin{theorem}\label{thm:main}
  Formula systems and alternating graph automata are of equal expressive power.
\end{theorem}

\section{Summary and Related Work}\label{s:concl}

We have introduced systems of graph formulas which extend the graph formulas
of~\cite{dhm-icgt25b} by an element of recursion, using a mechanism related
to the systems of language equations introduced by Mezei and Wright~\cite{Mezei-Wright:67}. Our main result is that these
formula systems have the same expressive power as the alternating graph
automata proposed in~\cite{dhm-icgt25a}. By results proved in the latter
paper this implies that the graph languages that can be specified by
formula systems are in PSPACE and include some PSPACE-complete languages.
Since we know from~\cite{dhm-icgt25b} that graph formulas on their own
(without variables) can only specify languages in the polynomial hierarchy PH,
it follows that formula systems are more powerful than graph formulas unless
$\text{PSPACE}=\text{PH}$.

Our notion of formula systems is related to the extension of the
widely known nested graph conditions of Rensink, Habel, and Pennemann
\cite{Rensink:04a,Habel-Pennemann:05,Pennemann:09} by Flick, termed
``recursively nested''~\cite{Flick:16}, which can specify certain non-local graph conditions
such as being acyclic or a tree. However, they cannot specify
conditions whose definition, e.g., demands the existence of certain cycle-free or
disjoint paths.
Hamiltonicity, for example, cannot be specified by recursively nested
conditions, which is possible with our formula systems. In fact, Hamiltonicity
can even be specified by a single graph formula without any variable, as shown 
in~\cite{dhm-icgt25b}.

  

While there has been work on finite graph-processing automata 
\cite{Witt:81,Kaminski-Pinter:92,Bozapalidis-Kalampakas:08,Kalampakas:11,%
Blume-Bruggink-Friedrich-Koenig:13,Blume:14,Bruggink-Koenig:18},
we are only aware of two papers on \emph{alternating}
graph automata:
\begin{enumerate*}[label=(\arabic*)]
\item The automata of Bruggink \emph{et al.}
  \cite{Bruggink-Huelsbusch-Koenig:12} appear to be
  weaker than ours, as discussed
  in~\cite[Example~5]{dhm-icgt25a}.
\item The automata of Brandenburg and
  Skodinis~\cite{Brandenburg-Skodinis:05} allow to specify graph
  languages that can be defined by node replacement
  \cite{Engelfriet-Rozenberg:97}, more precisely languages of
  undirected node-labelled graphs defined by boundary NCE grammars. In
  a universal configuration $(q,G)$ of their automata, with ongoing
  transitions $(q,\Gamma_1,q_1) \dots (q,\Gamma_k,q_k)$, cutting the
  occurrences of $\Gamma_1$ to $\Gamma_k$ off $G$ must result in
  pairwise unconnected remainder graphs $R_1 \dots R_k$ so that the
  automaton can proceed with configurations $(q_1,R_1)$ to $(q_k,R_k)$
  in parallel. This seems rather obscure, since such steps require a
  global check of the remainder graph, which contradicts the common
  understanding of automata.
\end{enumerate*}

%

Monadic second-order (MSO) logic on graphs is an extremely
well-studied formalism for specifying graph languages;
see~\cite{Courcelle-Engelfriet:12} and the multitude of references therein.
While also being a logic formalism, our formula systems are fundamentally
different: MSO logic is an instance of predicate logic whereas the formulas
introduced in \cite{dhm-icgt25b} and extended in the current paper are based
on the idea of analyzing a graph by repeatedly matching and cutting off
(frontal) subgraphs. Investigating the relation between the two formalisms
could be an interesting avenue for future work, our current conjecture being
that they are incomparable with respect to their expressive power.


\paragraph{Acknowledgments.} We thank the anonymous reviewers for their
useful comments, which helped to improve the paper. 

\bibliographystyle{eptcs}
\bibliography{References}

\end{document}